%% file: FedCG.tex
\documentclass[conference]{IEEEtran}
\IEEEoverridecommandlockouts
\usepackage{cite}
\usepackage{amsmath,amssymb,amsfonts}
\usepackage{graphicx}
\usepackage{textcomp}
\usepackage{xcolor}
\usepackage{bbding}

\usepackage{bbm}
\usepackage{bm}
\usepackage{booktabs} 
\usepackage{pifont}
\usepackage{balance}
\usepackage{subfigure}
\usepackage{color}
\usepackage{ifpdf}
\usepackage{epsfig}
\usepackage{latexsym}
\usepackage{paralist}
\usepackage{comment}
\usepackage{xspace}
\usepackage{mathrsfs}
\usepackage{epsf}
\usepackage{setspace}
\usepackage{color}
\usepackage{caption}
\usepackage{mathtools}
\usepackage{multirow}
\usepackage[noend]{algpseudocode}
\usepackage{amssymb}
\usepackage{url,epsfig,array}
\usepackage{leftidx}
\usepackage{epstopdf}
\usepackage{footmisc}
\usepackage{float}
\usepackage{amsthm}
\usepackage[ruled,linesnumbered,vlined]{algorithm2e}
\def\BibTeX{{\rm B\kern-.05em{\sc i\kern-.025em b}\kern-.08em
    T\kern-.1667em\lower.7ex\hbox{E}\kern-.125emX}}
\def\ie{\textit{i.e.}\xspace}

\def\eg{\textit{e.g.}\xspace}

\newtheorem{theorem}{Theorem}

\newtheorem{lemma}{Lemma}

\newtheorem{assumption}{Assumption}

\makeatletter
\renewcommand{\maketag@@@}[1]{\hbox{\m@th\normalsize\normalfont#1}}%
\makeatother

\begin{document}

\title{
	Adaptive Control of Client Selection and Gradient Compression for Efficient Federated Learning
}

\author{\IEEEauthorblockN{Zhida Jiang$^{1,2}$  \ \ Yang Xu$^{*1,2}$ \ \  Hongli Xu$^{*1,2}$  \ \ Zhiyuan Wang$^{1,2}$ \  \   Chen Qian$^3$ }
	\IEEEauthorblockA{
		$^1$School of Computer Science and Technology, University of Science and Technology of China\\
		$^2$Suzhou Institute for Advanced Research, University of Science and Technology of China\\
		$^3$Department of Computer Science and Engineering, Jack Baskin School of Engineering, University of California, Santa Cru\\
} }

\maketitle

\begin{abstract}
	\input{content/abstract.tex}
\end{abstract}


\section{Introduction}\label{sec:intro}
\input{content/intro.tex}

\section{Related Work}\label{sec:related}
\input{content/related.tex}

\section{Preliminaries and Problem Formulation}\label{sec:prelim}
\input{content/prelim.tex}

\section{Algorithm Design}\label{sec:algorithm}
\input{content/algorithm.tex}

\section{Performance Evaluation}\label{sec:evaluation}

\input{content/evaluation.tex}

\section{Conclusion}\label{sec:conclusion}
\input{content/conclusion.tex}

\bibliographystyle{IEEEtran}
\bibliography{content/refs}

\end{document}

%% file: content/abstract.tex
Federated learning (FL) allows multiple clients cooperatively train models without disclosing local data.
However, the existing works fail to address all these practical concerns in FL: limited communication resources, dynamic network conditions and heterogeneous client properties, which slow down the convergence of FL.
To tackle the above challenges, we propose a heterogeneity-aware FL framework, called FedCG, with adaptive client selection and gradient compression.
Specifically, the parameter server (PS) selects a representative client subset considering statistical heterogeneity and sends the global model to them.
After local training, these selected clients upload compressed model updates matching their capabilities to the PS for aggregation, which significantly alleviates the communication load and mitigates the straggler effect.
We theoretically analyze the impact of both client selection and gradient compression on convergence performance.
Guided by the derived convergence rate, we develop an iteration-based algorithm to jointly optimize client selection and compression ratio decision using submodular maximization and linear programming.
Extensive experiments on both real-world prototypes and simulations show that FedCG can provide up to 5.3$\times$ speedup compared to other methods. 

%% file: content/intro.tex
Recently, federated learning (FL) \cite{mcmahan2017communication} as a novel distributed machine learning paradigm has attracted a lot of attention.
In FL, training data are distributed across a large number of edge devices, such as mobile phones, personal computers, or smart home devices.
Under the orchestration of the parameter server (PS), these devices (\ie, clients) cooperatively train a global inference model without sharing raw data, which efficiently leverages local computing resources of edge devices and addresses data privacy concerns.
With the technical advantages and implemental feasibilities, FL has been applied in a variety of applications, such as next word prediction, extended reality, and smart manufacturing \cite{kairouz2021advances}.

Despite its practical effectiveness, there are several key challenges unique to the FL setting that make it difficult to train high-quality models.
\textit{(1) Limited communication resources.} 
Since the clients participating in FL need to communicate with the PS iteratively over bandwidth-limited networks, the communication cost is prohibitive and forms a huge impediment to FL’s viability, especially when training modern models with millions of parameters \cite{xu2021deepreduce}.
\textit{(2) Dynamic network conditions.} Owing to link instability and bandwidth competition, the communication conditions of wireless channels may fluctuate over time, resulting in dynamics of available bandwidth \cite{cui2021optimal}. For example, a user’s smartphone may be allocated higher bandwidth when transmitting model updates at night than during the day.
\textit{(3) Heterogeneous client properties.} The heterogeneity of the clients usually includes capability heterogeneity and statistical heterogeneity.
The clients may be equipped with different computing chips and located in diverse regions, thus their capabilities vary significantly \cite{nishio2019client}. The stragglers will delay the aggregation step and make the training process inefficient.
Besides, due to different user preferences and contexts, local data on each client are not independent and identically distributed (non-IID).
For instance, the images collected by the cameras reflect the demographics of each camera's location.
Heterogeneous statistical data will bring the biases in training and eventually cause an accuracy degradation of FL \cite{wang2020optimizing}.

To improve communication efficiency of FL, a natural solution is to reduce the size of transmitted payload or select only a fraction of clients participating in training.
The existing works have adopted quantization \cite{tang2018communication,alistarh2017qsgd,wen2017terngrad,bernstein2018signsgd,cui2021optimal} or sparsification \cite{aji2017sparse,stich2018sparsified,sattler2019robust,abdelmoniem2021dc2,DBLP:conf/iclr/LinHM0D18,wangni2018gradient,han2020adaptive,xu2021deepreduce} techniques to relax the communication load.
But these compression algorithms often assign fixed or identical compression ratios to all clients, which are agnostic to the capability heterogeneity and thereby result in considerable completion time lags.
Besides, these compression schemes \cite{li2021talk} do not take statistical heterogeneity into account, deteriorating training efficiency in the presence of non-IID data.
Another line of studies aims to design client selection (or client sampling) schemes based on heterogeneous client properties \cite{cho2020client,mohammed2020budgeted,chen2020optimal,balakrishnan2022diverse,wang2020optimizing,tang2021fedgp,rizk2021optimal,nishio2019client,perazzone2022communication,shi2020joint,chai2020tifl,jin2020resource}, most of which lack joint consideration of capability and statistical heterogeneity.
Although some works take two types of heterogeneity into account \cite{luo2021tackling}, the derived client selection probabilities are fixed during the training process, which cannot adapt to network dynamics.

In summary, most prior works fail to address all the aforementioned challenges, thereby hindering efficient FL.
This motivates us to study the following question: \textit{how to enhance FL by simultaneously addressing the challenges of communication efficiency, network dynamics and client heterogeneity?}

To tackle this problem, we propose a heterogeneity-aware FL framework, called FedCG (Federated Learning with \textbf{C}lient selection and \textbf{G}radient compression).
At each round, the PS selects a diverse subset of clients that carry representative gradient information and then sends the global model to the selected clients.
After local training, these clients adopt gradient compression to further boost communication efficiency.
FedCG adaptively assigns appropriate compression ratios to selected clients based on their heterogeneous and time-varying capabilities.
In this way, each client uploads compressed model updates matching its capabilities to the PS.
Finally, the PS aggregates the model updates to obtain the latest global model.
Under this framework, our advantages are reflected in two aspects.
On one hand, we select a representative client subset such that their aggregated model updates approximate full client aggregation \cite{balakrishnan2022diverse}.
By encouraging diversity in client selection, FedCG can effectively reduce redundant communication and promote fairness, which modulates the skew introduced by non-IID data.
On the other hand, different compression ratios will adapt to dynamic network conditions and heterogeneous capabilities, which contributes to mitigating the straggler effect and thus accelerates the training process.

More importantly, instead of directly combining client selection and gradient compression, we highlight that their decisions are interacted and demonstrate the need for joint optimization.
Specifically, the compression ratios should be adapted to the heterogeneous capabilities of the selected clients.
Correspondingly, client selection is also bound up with the degree of gradient compression. Selecting clients with over-compressed gradients will impede convergence.
As a result, the naive combination of existing client selection and gradient compression schemes cannot adequately address the key challenges of FL and may degrade training performance, which is empirically verified in Section \ref{sec:evaluation}.

However, jointly optimizing client selection and compression ratio is non-trivial for the following reasons.
\textit{Firstly}, the quantitative relationship between client selection, gradient compression and model convergence is unclear.
\textit{Secondly}, it is difficult to determine the proper compression ratios to achieve a delicate trade-off between resource overhead and model accuracy. Things will get even worse while considering the capability heterogeneity across different clients.
\textit{Thirdly}, the tightly coupled problem of client selection and compression ratio decision adds additional challenges to algorithm design.
In light of the above discussion, we state the key contributions of this paper as follows:
\begin{itemize}
	\setlength{\itemsep}{2pt}
	\setlength{\parsep}{2pt}
	\setlength{\parskip}{2pt}
	\item We propose a novel FL framework, called FedCG, which addresses the challenges of communication efficiency, network dynamics and client heterogeneity by adaptive client selection and gradient compression. We theoretically analyze the impact of client selection and gradient compression on convergence performance.
	
	\item Guided by the convergence analysis, we apply submodular maximization to select diverse clients, and determine different compression ratios for heterogeneous clients to achieve the trade-off between overhead and accuracy.
	We develop an iteration-based algorithm to jointly optimize client selection and compression ratio decision for the tightly coupled problem.
	
	\item We evaluate the performance of our proposed framework on both a hardware platform and a simulated environment.
	Extensive experimental results demonstrate that for both convex and non-convex machine learning models, FedCG can provide up to 5.3$\times$ speedup compared to state-of-the-art methods.
\end{itemize}

The remainder of this paper is organized as follows. Section
\ref{sec:related} reviews related work. Section \ref{sec:prelim} introduces our proposed framework and formulates the optimization problem.
Section \ref{sec:convergence} provides the convergence analysis of FedCG. Section \ref{sec:algorithm} designs a joint optimization algorithm for client selection and compression ratio decision. Section \ref{sec:evaluation} presents experimental results, and finally Section \ref{sec:conclusion} concludes the paper.

%% file: content/related.tex
In FL, the iterative communication between the PS and clients will incur considerable costs, particularly when the underlying model is of high complexity \cite{xu2021deepreduce}.
To this end, various works have been devoted to improving communication efficiency by reducing the size of transmitted models/gradients or selecting a subset of clients.
On one hand, compression techniques have been adopted to alleviate the transmission burden, including quantization \cite{tang2018communication,alistarh2017qsgd,wen2017terngrad,bernstein2018signsgd,cui2021optimal} and sparsification \cite{aji2017sparse,stich2018sparsified,sattler2019robust,abdelmoniem2021dc2,DBLP:conf/iclr/LinHM0D18,wangni2018gradient,han2020adaptive,xu2021deepreduce}.
The quantization based methods \cite{tang2018communication,alistarh2017qsgd,wen2017terngrad,bernstein2018signsgd} aim to represent each element with fewer bits, such as QSGD \cite{alistarh2017qsgd} and signSGD \cite{bernstein2018signsgd}.
Optimal compression ratio allocation for quantization is considered in \cite{cui2021optimal}.
Other studies \cite{aji2017sparse,stich2018sparsified,sattler2019robust,abdelmoniem2021dc2,DBLP:conf/iclr/LinHM0D18,wangni2018gradient,han2020adaptive,xu2021deepreduce} apply sparsification to transmit a small subset of gradients so that the communication overhead can be reduced dramatically.
However, the aforementioned works assign identical compression ratios to heterogeneous clients and thus the stragglers with poor channel conditions will become the bottleneck of model training.
The authors in \cite{li2021talk} provide client-specific compression schemes according to communication heterogeneity. However, they do not consider statistical heterogeneity and thus exhibit poor performance in the presence of non-IID data, in terms of model accuracy and convergence rate.

On the other hand, client selection plays a critical role in FL and has been extensively studied in previous works.
In the common implementation, clients are selected uniformly at random or proportional to local dataset size \cite{mcmahan2017communication,DBLP:conf/iclr/LiHYWZ20}, which results in poor training performance and long latency due to non-IID data and capability heterogeneity \cite{luo2021tackling}.
Considering the statistical property, some sampling methods have investigated different criteria to evaluate the importance of clients, such as local loss \cite{cho2020client}, test accuracy \cite{mohammed2020budgeted}, model updates \cite{chen2020optimal,balakrishnan2022diverse,wang2020optimizing}, client correlations \cite{tang2021fedgp}, and local data variability \cite{rizk2021optimal}.
However, the above strategies ignore the heterogeneity of clients' capabilities and may suffer from the straggler effect.
Other studies \cite{nishio2019client,perazzone2022communication,shi2020joint,chai2020tifl,jin2020resource} have designed client selection schemes that tackle heterogeneous system resources for fast convergence, but non-IID data still hurt the model accuracy.
A very recent work \cite{luo2021tackling} optimizes client selection probabilities while accounting for both data and capability heterogeneity. 
In this solution, the exchange of complete models incurs exorbitant communication cost and the obtained probabilities cannot be adaptively adjusted as training progresses, which ignores time-varying network conditions and thus exhibits less flexibility.
Compared with the prior works, FedCG can \textit{simultaneously} cope with the challenges of communication efficiency, network dynamics and client heterogeneity by joint optimization of client selection and gradient compression.

%% file: content/prelim.tex

\subsection{Federated Learning Basics}
The goal of FL is to train a high-quality model through a loose federation of clients, which is coordinated by the PS.
We suppose there is a set $\mathcal{N}=\{1,2,...,N\}$ of clients participating in FL.
Each client $n\in \mathcal{N}$ has its local dataset $\mathcal{D}_n$ with the size of $|\mathcal{D}_n|$.
The local loss function of client $n$ on the collection of data samples is defined as:
\vspace{-1mm}
\begin{equation}\small
F_n(\mathbf{x})=\frac{1}{|\mathcal{D}_n|}\sum_{\xi\in \mathcal{D}_n} f_n(\mathbf{x};\xi),
\vspace{-2mm}
\end{equation}
where $\mathbf{x}$ is the model parameter vector and $f_n(\mathbf{x};\xi)$ is the loss function calculated by a specific sample $\xi$.
FL seeks to minimize the global loss function $F(\mathbf{x})$, which translates into the following optimization problem:
\vspace{-1.5mm}
\begin{equation}\label{eq1}\small
	\min \limits_{\mathbf{x}} \  F(\mathbf{x})= \sum_{n=1}^{N} p_nF_n(\mathbf{x}),
	\vspace{-2mm}
\end{equation}
where $p_n$ represents the weight of client $n$ with $\sum_{n=1}^{N}p_n=1$ and can be set to $p_n=\frac{|\mathcal{D}_n|}{\sum_{i=1}^{N}|\mathcal{D}_{i}|}$.

As a primitive implementation and the most commonly studied FL algorithm, \textit{FedAvg} \cite{mcmahan2017communication} has been proposed to solve the problem in Eq. (\ref{eq1}).
Specifically, the optimization process consists of multiple communication rounds. At each round $k\in\left\{0,1,...,K-1\right\}$, the PS randomly selects $M$ clients and sends the global model $\mathbf{x}^k$ to the set of selected clients $\mathcal{M}^k\subseteq\mathcal{N}$.
By setting $\mathbf{x}^{k,0}_n=\mathbf{x}^k$, each client $n$ independently trains the local model for $H$ iterations:
\vspace{-0.8mm}
\begin{equation}\small
	\mathbf{x}^{k,j+1}_n = \mathbf{x}^{k,j}_n - \eta_k\nabla F_n(\mathbf{x}^{k,j}_n;\xi^{k,j}_n),\ j = 0,1,...,H-1,
	\vspace{-1mm}
\end{equation}
where $\eta_k$ is the learning rate at round $k$, and $\xi^{k,j}_n$ is the sample selected by client $n$ for local iteration $j$.
After local training, each client $n$ sends the model updates $\mathbf{G}^k_n=\sum_{j=0}^{H-1}\nabla F_n(\mathbf{x}^{k,j}_n;\xi^{k,j}_n)$ to the PS for global aggregation.
However, unlike in a cloud data center, FedAvg might face a few fundamental challenges while training models on edge devices, such as limited communication resources, dynamic network conditions and heterogeneous client properties.

\subsection{Heterogeneity-Aware Federated Learning Framework}
To address these challenges, we propose a heterogeneity-aware FL framework, called FedCG. As shown in Alg. \ref{training_process}, the training process of our framework includes $K$ rounds, and each round consists of the following phases.
\vspace{-0.9mm}
\begin{itemize}
	\setlength{\itemsep}{0pt}
	\setlength{\parsep}{0pt}
	\setlength{\parskip}{0pt}
	\item At the beginning of round $k$, FedCG adaptively selects a diverse subset of clients $\mathcal{M}^k$ considering statistical heterogeneity and determines different compression ratios for selected clients according to heterogeneous and time-varying capabilities.
	Then the PS sends the global model $\mathbf{x}^k$ and compression ratio $\theta_n^k$ to each client $n\in\mathcal{M}^k$.
	
	\item Each client $n$ updates the received model over its local dataset for $H$ iterations. Based on compression ratio $\theta_n^k$, the client $n$ compresses the original model updates $\mathbf{G}^k_n$ to obtain $\mathbf{\tilde{G}}^k_n$.
	Then the compressed model updates $\mathbf{\tilde{G}}^k_n$ that fit the capabilities of client $n$ are uploaded to the PS.
	
	\item Upon receiving all model updates from selected subset, the PS obtains a new global model $\mathbf{x}^{k+1}$  by aggregating compressed model updates and starts the next round.
\end{itemize}
\vspace{-1mm}
Next, we detail the two main innovations of our framework, \ie, client selection and gradient compression.

\begin{algorithm}[!t]
	\setstretch{1.1}
	\caption{Training process of FedCG}
	\label{training_process}
	\For{Each round $k= 0,1,...,K-1$}
	{
		The PS selects a diverse subset of clients $\mathcal{M}^k$ with $\left|\mathcal{M}^k\right| =M$\;
		The PS determines differnet compression raio $\theta_n^k$ for each client $n\in\mathcal{M}^k$\;
		The PS sends the current global model $\mathbf{x}^k$ and compression raio $\theta_n^k$ to each client $n\in\mathcal{M}^k$\;
		\For{Each client $n \in \mathcal{M}^k$ in parallel}
		{
			$\mathbf{x}^{k,0}_n=\mathbf{x}^k$\;
			\For{Each local iteration $j= 0,1,...,H-1$}
			{
				$\mathbf{x}^{k,j+1}_n = \mathbf{x}^{k,j}_n - \eta_k\nabla F_n(\mathbf{x}^{k,j}_n;\xi^{k,j}_n)$\;
			}
			Compress the model updates $\mathbf{G}^k_n$ to obtain $\mathbf{\tilde{G}}^k_n$ according to compression raio $\theta^k_n$\;
			Upload $\mathbf{\tilde{G}}^k_n$ to the PS\;
		}
		The PS updates the global model $\mathbf{x}^{k+1} = \mathbf{x}^k - \frac{\eta_k}{M}\sum_{n\in\mathcal{M}^k}\mathbf{\tilde{G}}^k_n$\;
	}
\end{algorithm}

\textbf{(1) Client Selection.}
Considering limited communication bandwidth and client availability, we select a fraction of clients to participate in training, which effectively reduces communication overhead.
However, the clients located in geographically distinct regions generate data from different distributions (\ie, non-IID) in practice.
Many clients may provide similar and redundant gradient information for aggregation, which cannot reflect the true data distribution in the global view.
Selecting such clients will waste resources and cause the global model to be biased towards certain clients, thus exacerbating the negative impact of non-IID data on training performance.
To this end, we introduce diversity to client selection and select representative clients out of the whole while adhering to resource constraints \cite{balakrishnan2022diverse}.
We expect to find a diverse subset of clients $\mathcal{M}^k$ whose aggregated model updates approximate the (logically) aggregated updates of all clients.
By encouraging diversity in model updates, we reduce redundant communication and increase the impact of under-represented clients that contribute different information, thereby promoting fairness.
In this way, FedCG will counterbalance the bias introduced by non-IID data and speed up convergence.

\textbf{(2) Gradient Compression.}
Gradient compression is another commonly adopted solution to alleviate network pressure due to its practicality and substantial bandwidth efficiency.
Among previous gradient compression techniques, Top-k is a promising compression operator with empirical and theoretical studies \cite{DBLP:conf/iclr/LinHM0D18,shi2019convergence}.
Only the gradients with larger absolute values are required to transmit for global aggregation, which can sparsify the local gradients to only 0.1\% density without impairing model convergence or accuracy \cite{shi2019convergence}.
Therefore, we adopt Top-k sparsification in this paper due to its efficiency and simplicity.
It is worth mentioning that other compression operators (\eg, Random-k sparsification \cite{stich2018sparsified} and quantization \cite{alistarh2017qsgd}) can also be compatible with our framework.
Based on compression ratio $\theta_n^k$, the client $n$ selects the gradient elements with larger absolute values from the original model updates $\mathbf{G}^k_n$, and zero-out other unselected gradient elements to obtain compressed model updates $\mathbf{\tilde{G}}^k_n$.
Moreover, we apply the error compensation mechanism \cite{li2021talk} onto FedCG, which is widely used along with compression to further improve training performance.
Error compensation accumulates the error from only uploading compressed gradients, thereby ensuring that all elements of the full gradient have a chance to be aggregated.

The compression ratio $\theta_n^k$ can be regarded as a measure of the sparsity, where a smaller $\theta_n^k$ corresponds to a more sparse vector and requires less communication and vice versa.
More importantly, unlike the existing works unifying the sparsity levels of clients, FedCG assigns \textit{different} compression ratios to the selected clients $\mathcal{M}^k$ considering their heterogeneous and time-varying capabilities.
Specifically, the clients with excellent capabilities (\ie, short completion time) are expected to adopt slight gradient compression, while the others with poor capabilities (\ie, long completion time) should compress the gradients more aggressively. As a result, the selected clients will achieve approximately identical per-round completion time.
Adaptive gradient sparsification dramatically saves the communication cost and mitigates the impact of stragglers, which provides substantial benefits in improving training efficiency, particularly in the resource-limited and heterogeneous wireless environments envisioned for FL.

\subsection{Problem Formulation}
We define the joint optimization problem of client selection and compression ratio decision as below.
Let $T_{n,cmp}^k$ denote the computation time required for client $n$ to perform one local iteration.
At round $k$, the local training time of client $n$ is $HT_{n,cmp}^k$ \cite{luo2021cost}, where $H$ is the number of local iterations between two consecutive global synchronizations.
Following the prior works \cite{DBLP:conf/iclr/LinHM0D18}, we only consider the uplink communication, since the downlink speed in FL is much faster compared with the uplink and the parameter download time is negligible \cite{zhan2020incentive}.
The communication time of client $n$ at round $k$ can be formulated as:
\vspace{-1mm}
\begin{equation}\small
	T_{n,com}^k=\frac{\theta_n^kR}{C_n^k},
	\vspace{-1mm}
\end{equation}
where $R$ represents the size of original (uncompressed) model updates and $C_n^k$ represents the upload speed of client $n$ at round $k$.
The upload speed changes dynamically as the training progresses.
For client $n$, the total time $T_n^k$ of local training and transmitting parameters at round $k$ is expressed as:
\vspace{-1mm}
\begin{equation}\small
	T_n^k=HT_{n,cmp}^k+T_{n,com}^k.
	\vspace{-1mm}
\end{equation}
In the synchronous FL, the per-round time is determined by the “slowest” one among the selected clients. The completion time of round $k$ is defined as:
\vspace{-1mm}
\begin{equation}\small
	T^k=\underset{n\in\mathcal{M}^k}{\max}\ T_n^k.
	\vspace{-1mm}
\end{equation}

We aim to select the clients involved in FL and determine the compression ratios for those selected clients.
The optimization problem can be formulated as:
\vspace{-1mm}
\begin{equation}\small
	\min\ F(\mathbf{x}^{K})\nonumber
\end{equation}
\vspace{-4mm}
\begin{equation}\label{equ:mini}\small
	s.t.\\
	\begin{cases}
		\sum_{k=0}^{K-1}T^k<T  \\
		\left|\mathcal{M}^k\right| =M, &\forall k \\
		0<\theta_n^k\leq1, &\forall n, \forall k \\
	\end{cases}
\end{equation} 
The first inequality guarantees the resource constraint where $T$ denotes the total time budget for given $K$.
The second set of inequalities indicates that the PS selects $M$ clients participating in training at each round.
The third set of inequalities bounds the feasible range of compression ratios.
The objective of the optimization problem is to minimize the loss function $F(\mathbf{x}^{K})$ of model training given the resource constraint.

It is worth noting that our formulation can be extended to other “costs” beyond the completion time (\eg, energy consumption) as well.
In fact, it is non-trivial to directly solve the problem in Eq. (\ref{equ:mini}) due to the unclear convergence relationship, accuracy-overhead trade-off, and tightly coupled nature.
In the following section, we derive a tractable convergence rate to indicate how the selected clients and compression ratios affect the final convergence. On this basis, we develop a joint optimization algorithm to solve the coupled problem.

\section{Convergence Analysis}\label{sec:convergence}
In this section, we provide the convergence analysis of the proposed framework. We first state the following assumptions on the local loss functions.

\begin{assumption}\label{assum1}
	\vspace{-2mm}
	$F_1,F_2,...,F_N$ are all $L$-smooth, \ie, given $\mathbf{x}$ and $\mathbf{y}$, $F_n(\mathbf{x})\le F_n(\mathbf{y})+(\mathbf{x}-\mathbf{y})^T\nabla F_n(\mathbf{y})+\frac{L}{2} \left \| \mathbf{x}-\mathbf{y} \right \|^2$.\\
	\vspace{-3mm}
\end{assumption}
\vspace{-4mm}
\begin{assumption}\label{assum2}
	$F_1,F_2,...,F_N$ are all $\mu$-strongly convex, \ie, given $\mathbf{x}$ and $\mathbf{y}$, $F_n(\mathbf{x})\ge F_n(\mathbf{y})+(\mathbf{x}-\mathbf{y})^T\nabla F_n(\mathbf{y})+\frac{\mu}{2} \left \| \mathbf{x}-\mathbf{y} \right \|^2$.
\end{assumption}
\vspace{-4mm}
\begin{assumption}\label{assum3}
	The variance of the stochastic gradients on random data samples is bounded, \ie, $\mathbb{E} [  \| \nabla F_n(\mathbf{x}^{k,j}_n;\xi^{k,j}_n)-\nabla F_n(\mathbf{x}^{k,j}_n) \|^2 ] \le \sigma^2 , \forall n, \forall j, \forall k$.
\end{assumption}
\vspace{-4mm}
\begin{assumption}\label{assum4}
	The stochastic gradients on random data samples are uniformly bounded, \ie, $\| \nabla F_n(\mathbf{x}^{k,j}_n;\xi^{k,j}_n)\|^2 \le G^2 , \forall n, \forall j, \forall k $.
\end{assumption}
\vspace{-2mm}
These assumptions hold for typical FL models and are common in the convergence analysis literature \cite{balakrishnan2022diverse,DBLP:conf/iclr/LiHYWZ20,cui2021optimal,luo2021tackling,wang2022accelerating}. Although our convergence analysis focuses on strong convex problems, the experimental results demonstrate that proposed framework also works well for non-convex learning problems. Furthermore, we use $\Gamma =F^{*}- \sum_{n=1}^{N}p_nF_n^{*} $ to quantize the degree of non-IID data distribution on clients, where $F^{*}$ and $F_n^{*}$ are the optimal values of $F$ and $F_n$, respectively.
If the data across clients follow IID, then $\Gamma$ obviously goes to zero as the number of samples grows. 
Inspired by \cite{DBLP:conf/iclr/LiHYWZ20}, we flatten local iterations at each communication round and use $\mathbf{y}^{k,j+1}_n = \mathbf{x}^{k,j}_n - \eta_k\nabla F_n(\mathbf{x}^{k,j}_n;\xi^{k,j}_n)$ to represent the result of a local iteration on client $n$. If $j + 1 < H$, $\mathbf{x}^{k,j+1}_n = \mathbf{y}^{k,j+1}_n$; otherwise, $\mathbf{x}_n^{k+1,0} = \mathbf{x}_n^{k,0} - \frac{\eta_k}{M}\sum_{n\in\mathcal{M}^k}\mathbf{\tilde{G}}^k_n$ and $\mathbf{y}_n^{k+1,0} = \mathbf{y}_n^{k,H}$. Moreover, we define $\bar{\mathbf{x}}^{k,j} =\sum_{n=1}^{N} p_n\mathbf{x}^{k,j}_n$, $\bar{\mathbf{y}}^{k,j} =\sum_{n=1}^{N} p_n\mathbf{y}^{k,j}_n$, $\bar{\mathbf{x}}^{k}=\bar{\mathbf{x}}^{k,0}$ and $\bar{\mathbf{y}}^{k}=\bar{\mathbf{y}}^{k,0}$.

At round $k$, we need to find a subset $\mathcal{M}^k$ of clients whose aggregated gradients can approximate the full gradients over all the $N$ clients. We assume that there is a mapping $\pi^k: \mathcal{N}\to\mathcal{M}^k$ such that the gradients from client $n\in\mathcal{N}$ can be approximated by the gradients from a selected client $\pi^k(n)\in\mathcal{M}^k$.
Let $\mathcal{A}_n^k=\{i\in\mathcal{N}|\pi^k(i)=n\}$ be the set of clients approximated by client $n\in\mathcal{M}^k$ and $\gamma_n^k=|\mathcal{A}_n^k|$.
Then we define the approximation error at round $k$ as:
\begin{equation}\small
\alpha_k = \left \| \frac{1}{N}\sum_{n\in\mathcal{M}^k}\gamma_n^k\nabla F_n(\mathbf{y}_n^{k,0})-\frac{1}{N}\sum_{n\in\mathcal{N}}\nabla F_n(\mathbf{y}_n^{k,0})  \right \| ,
\end{equation} 
which is used to characterize how well the aggregated gradients of selected client subset $\mathcal{M}^k$ approximate the full gradients.
Besides, we define the compression error as:
\vspace{-1mm}
\begin{equation}\small
	\beta_n^k = \|\mathbf{\tilde{G}}^k_n-\mathbf{G}^k_n \|,
	\vspace{-1mm}
\end{equation}
which indicates the difference between the compressed gradients $\mathbf{\tilde{G}}^k_n$ and the original gradients $\mathbf{G}^k_n$ of client $n$ at round $k$.
The approximation error and compression error are related to client selection and compression ratio decision policies. Their impact on final accuracy will be quantified in the next theorem.

\vspace{-1mm}
\begin{lemma}\label{lemma1}
	Under Assumptions \ref{assum1}-\ref{assum4}, the proposed framework ensures
	
	\vspace{-4mm}
	\begin{small}
		\begin{align*}
			\|\bar{\mathbf{x}}^{k+1}-\bar{\mathbf{y}}^{k+1}\|\le LGH^2\eta_k^2 + \alpha_k H\eta_k+\frac{\eta_k}{M}\sum_{n\in\mathcal{M}^k}\beta_n^k.
		\end{align*}
	\end{small}
\end{lemma}
\vspace{-4mm}
\begin{proof}Since the clients only upload the model updates to the PS every $H$ local iterations, the same client subset $\mathcal{M}^k$ is adopted to approximate the full gradients at each iteration $j\in[0,H-1]$. 
With a slight abuse of notation, we have
	\begin{align}
		&\|\bar{\mathbf{x}}^{k+1}-\bar{\mathbf{y}}^{k+1}\| \nonumber\\
		=& \|(\bar{\mathbf{x}}^{k}- \frac{\eta_k}{M}\sum_{n\in\mathcal{M}^k}\mathbf{\tilde{G}}^k_n)-(\bar{\mathbf{x}}^{k}- \frac{\eta_k}{N}\sum_{n\in\mathcal{N}}\mathbf{G}^k_n)\|\nonumber\\
		\overset{(a)}{\le}&\sum_{j=0}^{H-1}\eta_k\|\frac{1}{N}\sum_{n\in\mathcal{M}^k}\gamma_n^k\nabla F_n(\mathbf{y}^{k,j}_n)- \frac{1}{N}\sum_{n\in\mathcal{N}}\nabla F_n(\mathbf{y}^{k,j}_n)\| \nonumber\\
		&+\frac{\eta_k}{M}\sum_{n\in\mathcal{M}^k}\|\mathbf{\tilde{G}}^k_n-\mathbf{G}^k_n \| \nonumber\\
		\overset{(b)}{\le}&\sum_{j=0}^{H-1}2LGj\eta_k^2 + \alpha_k H\eta_k + \frac{\eta_k}{M}\sum_{n\in\mathcal{M}^k}\|\mathbf{\tilde{G}}^k_n-\mathbf{G}^k_n \| \nonumber\\
		\le&LGH(H-1)\eta_k^2 + \alpha_k H\eta_k + \frac{\eta_k}{M}\sum_{n\in\mathcal{M}^k}\|\mathbf{\tilde{G}}^k_n-\mathbf{G}^k_n \| \nonumber\\
		\le&LGH^2\eta_k^2 + \alpha_k H\eta_k + \frac{\eta_k}{M}\sum_{n\in\mathcal{M}^k}\beta_n^k \nonumber
	\end{align}
where (a) follows from triangular inequality and (b) follows from the definition of approximation error.
\end{proof}

\begin{theorem}\label{theorem1}
Let Assumptions \ref{assum1}-\ref{assum4} and Lemma \ref{lemma1} hold, and $\mathbf{x}^{*}$ is the optimal model. We have the convergence rate:
	
	\vspace{-4mm}
	\begin{small}
		\begin{align*}
			\mathbb{E} [  \| \mathbf{x}^K-\mathbf{x}^{*} \|^2 ]  \le \mathcal{O}(\frac{1}{K})+\mathcal{O}(\alpha)+\mathcal{O}(\beta),
		\end{align*}
	\end{small}where $\alpha=\max\limits_{k}\{\alpha_k\}$ and $\beta=\max\limits_{n,k}\{\beta_n^k\}$.
\end{theorem}
\begin{proof}We analyze the interval between $\bar{\mathbf{x}}^{k+1}$ and $\mathbf{x}^{*}$:
	\begin{align}
		\mathbb{E} [  \| \bar{\mathbf{x}}^{k+1}-\mathbf{x}^{*} \|^2]
		=&\mathbb{E} [  \| \bar{\mathbf{x}}^{k+1}-\bar{\mathbf{y}}^{k+1}+\bar{\mathbf{y}}^{k+1}-\mathbf{x}^{*} \|^2] \nonumber\\
		=&\mathbb{E} [  \| \bar{\mathbf{x}}^{k+1}-\bar{\mathbf{y}}^{k+1}\|^2]+\mathbb{E} [  \|\bar{\mathbf{y}}^{k+1}-\mathbf{x}^{*} \|^2]\nonumber\\
		&+2\mathbb{E} [  \langle \bar{\mathbf{x}}^{k+1}-\bar{\mathbf{y}}^{k+1},\bar{\mathbf{y}}^{k+1}-\mathbf{x}^{*} \rangle] \label{eq50}
	\end{align}
	We know that $\bar{\mathbf{y}}^{k,j+1}$ satisfie the following properties \cite{DBLP:conf/iclr/LiHYWZ20}:
	\begin{align}\hspace{-1mm} \mathbb{E} [  \| \bar{\mathbf{y}}^{k,j+1}-\mathbf{x}^{*} \|^2 ]  \le (1-\eta_k\mu)\mathbb{E} [  \| \bar{\mathbf{x}}^{k,j}-\mathbf{x}^{*} \|^2 ]+\eta_k^2B \label{eq51}
	\end{align} 
	where $B=\sum_{n=1}^{N}p_n^2\sigma^2+6L\Gamma+8(H-1)^2G^2$.
	Substitute Eq. (\ref{eq51}) into Eq. (\ref{eq50}).
	For a decaying learning rate $\eta_k=\frac{\lambda}{k+\tau}$ where $\lambda>\frac{1}{\mu}$ and $\tau>0$ \cite{cui2021optimal}, we have
	\begin{align}
		&\mathbb{E} [  \| \bar{\mathbf{x}}^{k+1}-\mathbf{x}^{*} \|^2] \nonumber\\
		=&\mathbb{E} [  \| \bar{\mathbf{x}}^{k+1}-\bar{\mathbf{y}}^{k+1}\|^2]+(1-\eta_k\mu)\mathbb{E} [  \|  \bar{\mathbf{x}}^{k,H-1}-\mathbf{x}^{*} \|^2 ]  \nonumber+\eta_k^2B\nonumber\\
		&+2\mathbb{E} [  \langle \bar{\mathbf{x}}^{k+1}-\bar{\mathbf{y}}^{k+1},\bar{\mathbf{y}}^{k+1}-\mathbf{x}^{*} \rangle] \nonumber\\
		\overset{(a)}{\le}&(LGH^2\eta_k^2 + \alpha H\eta_k+\beta\eta_k)^2 +(1-\eta_k\mu)\mathbb{E} [  \|  \bar{\mathbf{x}}^{k,H-1}-\mathbf{x}^{*} \|^2 ]  \nonumber\\
		&+\eta_k^2B+ 2(LGH^2\eta_k^2 + \alpha H\eta_k+\beta\eta_k)\mathbb{E} [  \|\bar{\mathbf{y}}^{k+1}-\mathbf{x}^{*} \|]  \nonumber\\
		\le&(LGH^2\eta_k^2 + \alpha H\eta_k+\beta\eta_k)^2 +(1-\eta_k\mu)\mathbb{E} [  \|  \bar{\mathbf{x}}^{k}-\mathbf{x}^{*} \|^2 ]  \nonumber\\
		&+\eta_k^2HB+ 2(LGH^2\eta_k^2 + \alpha H\eta_k+\beta\eta_k)\mathbb{E} [  \|\bar{\mathbf{y}}^{k+1}-\mathbf{x}^{*} \|]  \nonumber\\
		\overset{(b)}{\le}&(1-\eta_k\mu)\mathbb{E} [  \|  \bar{\mathbf{x}}^{k}-\mathbf{x}^{*} \|^2 ]+ 2\eta_k\rho(\alpha H+\beta) \nonumber\\
		&+\eta_k^2((LGH^2\eta_k + \alpha H+\beta)^2+HB+2LGH^2\rho) \label{eq10}
	\end{align}
	where (a) follows from Lemma \ref{lemma1} and (b) follows since $\mathbb{E} [  \|\bar{\mathbf{y}}^{k+1}-\mathbf{x}^{*} \|]$ can be bounded by a constant $\rho$ \cite{balakrishnan2022diverse}.
	Based on Eq. (\ref{eq10}), the final convergence rate follows from \cite{mirzasoleiman2020coresets}.
\end{proof}

\textbf{Remark:} Theorem \ref{theorem1} reveals that the approximation error and compression error have a great impact on the convergence performance.
Ideally, when we select all clients to participate in training (\ie, $M=N$) and set the compression ratios of all clients as 1 (\ie, without compression) at each round, the approximate error and compression error become 0.
To maximize the final model accuracy for total $K$ rounds, we should minimize the approximation error and compression error under resource constraints.

%% file: content/algorithm.tex
In this section, we show how to leverage the derived convergence rate in Theorem 1 to obtain client selection and gradient compression policies for heterogeneous FL systems, which is the crucial design in FedCG.
We first introduce the overall joint optimization process (Section \ref{subsec:joint_optimization}) and then detail two core components of the proposed algorithm, \ie, client selection strategy (Section \ref{subsec:client_selection}) and compression ratio decision strategy (Section \ref{subsec:ratio_decision}), which are designed by minimizing the approximation error and compression error, respectively.

\subsection{Joint Optimization Process}\label{subsec:joint_optimization}
The key insight behind FedCG is that client selection and compression ratio decision interact with each other.
We cannot reach the optimal state by independently determining the client subset and compression ratios.
Therefore, this coupled property raises the necessity for joint optimization.
However, it is usually difficult to optimize both at the same time.
While if we fix one decision and then optimize the other, both of which are greatly simplified.
To this end, we propose an iteration-based algorithm to jointly optimize client selection and compression ratio decision for the tightly coupled problem.

As shown in Alg. \ref{joint_algorithm}, in each iteration, we first apply submodular maximization in Section \ref{subsec:client_selection} to select a diverse subset from candidate clients, thereby minimizing the approximation error (Line 4). The aggregated gradients of selected clients are a good approximation of the full gradients from all clients.
Then we determine appropriate compression ratios for these clients by solving the optimization problem (\ref{equ:ratio}) in Section \ref{subsec:ratio_decision} (Line 5).
If the derived solution contributes to the reduction of compression error, we update the current client subset and compression ratios (Lines 6-8).
Then, we find the client with the smallest compression ratio in the subset, and remove it from the candidate clients of the next iteration (Lines 9-10).
This design prevents the client with over-compressed gradients from participating in FL and affecting model accuracy.
The iterative heuristic terminates after $M$ iterations. We finally obtain the client subset $\mathcal{M}^{k}$ at round $k$ and the corresponding compression ratios $\theta_n^{k}$, $\forall n\in\mathcal{M}^{k}$ (Line 11).

Considering the coupled nature of the optimization problem, we apply an iteration-based algorithm to derive an efficient solution for client selection and compression ratio decision.
Starting with the initialization, the client subset and compression ratios of each round are optimized alternately via fixed-point iterations, thus optimizing the overall objective. It is worth noting that our algorithm
can be completed in $M$ iterations, which does not depend on the total number of clients (\ie, $N$).
Consequently, the algorithm overhead does not increase significantly with system scale and will be evaluated through the experiments in Section \ref{sec:evaluation}.

\begin{algorithm}[!t]
	\setstretch{1.1}
	\caption{Joint optimization algorithm at round $k$}
	\label{joint_algorithm}
	Initialize $\mathcal{M}^k=\emptyset$ and $\theta_n^{k}=0, \forall n$\;
	Initialize ${\mathcal{N}}'=\mathcal{N}$\;
	\For{Each iteration $i= 1,2,...,M$}
	{
		Select a diverse set of clients $\mathcal{M}^{k,i}$ from ${\mathcal{N}}'$ via submodular maximization in Section \ref{subsec:client_selection}\;
		Decide compression ratios $\theta_n^{k,i}$ for selected clients by solving optimization problem in Section \ref{subsec:ratio_decision}\;
		\If{$\sum_{n\in\mathcal{M}^{k,i}}\theta_n^{k,i}>\sum_{n\in\mathcal{M}^{k}}\theta_n^{k}$}
		{
			$\mathcal{M}^{k}\gets\mathcal{M}^{k,i}$\;
			$\theta_n^{k}\gets\theta_n^{k,i}$\;
		}
		${n}'=\arg\min_{n\in\mathcal{M}^{k,i}}\theta_n^{k,i}$\;
		${\mathcal{N}}'\gets {\mathcal{N}}'- \{{n}'\}$\;
	}
	\Return $\mathcal{M}^{k}$ and $ \theta_n^{k}$\;
\end{algorithm}

\subsection{Client Selection Strategy}\label{subsec:client_selection}
Based on theoretical analysis, we aim to select a subset of clients to minimize approximation error under resource constraints, thereby improving convergence performance.
The approximation error reflects how well the aggregated gradients of the selected clients approximate the gradients from all clients.
To this end, we introduce diversity to client selection so that the selected clients can be representative of all clients \cite{balakrishnan2022diverse}.
Submodular functions have been widely adopted to measure diversity \cite{mirzasoleiman2020coresets,minoux1978accelerated}. Formally, for any subset $\mathcal{S}\subseteq \mathcal{U}\subseteq \mathcal{N}$ and $z\in \mathcal{N}\setminus \mathcal{U}$, a set function $Q$ is submodular if $Q(\mathcal{S}\cup\{z\})-Q(\mathcal{S})\ge Q(\mathcal{U}\cup \{z\})-Q(\mathcal{U})$, which indicates $z$ is more valuable for a smaller set $\mathcal{S}$ than for a larger set $\mathcal{U}$.
The marginal gain of $z$ for a subset $\mathcal{S}$ is denoted as $Q(\mathcal{S}\cup\{z\})-Q(\mathcal{S})$.
All submodular functions have the diminishing return property, \ie, the marginal gain that an element brings to a subset diminishes as more elements are added to the subset. 
Thanks to the diminishing return property, maximizing submodular functions effectively promotes diversity and reduces the redundancy \cite{minoux1978accelerated}.

Inspired by the above facts, our algorithm minimizes the approximation error by applying submodular maximization to select diverse clients. 
Based on triangular inequality, we can derive an upper bound of the approximation error:

\vspace{-3mm}
\begin{small}
\begin{align}
	&\left \| \frac{1}{N}\sum_{n\in\mathcal{N}}\nabla F_n(\mathbf{y}_n^{k,0}) - \frac{1}{N}\sum_{n\in\mathcal{M}^k}\gamma_n^k\nabla F_n(\mathbf{y}_n^{k,0}) \right \| \nonumber\\
	\le&\frac{1}{N}\sum_{n\in\mathcal{N}} \left \| \nabla F_n(\mathbf{y}_n^{k,0})-\nabla F_{\pi^k(n)}(\mathbf{y}_{\pi^k(n)}^{k,0})  \right \|. \label{eq13}
\end{align}
\end{small}Eq. (\ref{eq13}) is minimized when the mapping $\pi^k$ assigns each $n\in\mathcal{N}$ to a client in $\mathcal{M}^k$ with the most gradient similarity:

\vspace{-3mm}
\begin{small}
\begin{align}
	&\left \| \frac{1}{N}\sum_{n\in\mathcal{N}}\nabla F_n(\mathbf{y}_n^{k,0}) - \frac{1}{N}\sum_{n\in\mathcal{M}^k}\gamma_n^k\nabla F_n(\mathbf{y}_n^{k,0}) \right \| \nonumber\\
	\le&\frac{1}{N}\sum_{n\in\mathcal{N}} \min_{i \in \mathcal{M}^k} \left \| \nabla F_n(\mathbf{y}_n^{k,0})-\nabla F_{i}(\mathbf{y}_{i}^{k,0})  \right \| = Q(\mathcal{M}^k). \label{eq14}
\end{align}
\end{small}Minimizing the upper bound $Q(\mathcal{M}^k)$ of the approximation error or maximizing $\bar{Q}(\mathcal{M}^k)$ (a constant minus its negation) is essentially equivalent to maximizing a well-known submodular function, \ie, facility location function \cite{cornuejols1977uncapacitated}. Considering the constraint $\left|\mathcal{M}^k\right| =M$, the approximation error minimization problem can be transformed into a submodular maximization problem under cardinality constraint, which is NP-hard \cite{mirzasoleiman2020coresets}.
Fortunately, the greedy algorithm has been proven to be effective in solving submodular maximization problem and provides $(1$$-$$e^{-1})$-approximation to the optimal solution \cite{nemhauser1978analysis}.

At round $k$, the greedy algorithm for minimizing the approximation error starts with an empty set $\mathcal{M}^k=\emptyset$. From the candidate set $\mathcal{N}$, the client $n\in\mathcal{N}\setminus \mathcal{M}^k$ with the largest marginal gain is constantly added to $\mathcal{M}^k$ until $\left|\mathcal{M}^k\right| =M$:
\begin{equation}\small
\mathcal{M}^k\gets \mathcal{M}^k\cup \{n^{*}\}, n^{*}=\mathop{\arg\max}\limits_{n\in\mathcal{N}\setminus \mathcal{M}^k}[\bar{Q} (\mathcal{M}^k \cup \{n\})-\bar{Q}(\mathcal{M}^k)].
\end{equation}
The computational complexity of the greedy algorithm is $\mathcal{O}(N\cdot M)$. However, in practice, the complexity can be reduced to $\mathcal{O}(N)$ using stochastic greedy algorithms \cite{mirzasoleiman2015lazier}, and further improved by lazy evaluation \cite{minoux1978accelerated} and distributed implementations \cite{mirzasoleiman2016fast}.
Consequently, the algorithm overhead can be negligible compared to the massive overhead for model training and transmission \cite{mirzasoleiman2020coresets}, which is empirically verified in Section \ref{sec:evaluation}.
Besides, it is infeasible for the PS to collect the gradients from all clients for marginal gain calculation.
For the clients whose gradients have not been collected at the current round, we estimate the marginal gain with their historical gradient information \cite{balakrishnan2022diverse}.
In a nutshell, we relate the gradient approximation to submodular maximization.
According to the marginal gain of the submodular function, we select the diverse subset of clients to minimize the approximation error between the estimated and the full gradients.

\subsection{Compression Ratio Decision Strategy}\label{subsec:ratio_decision}
We also need to determine different compression ratios for selected clients so as to minimize the compression error, which characterizes the difference between the compressed gradients and the original gradients.
The compression error satisfies the following contraction property \cite{stich2018sparsified}:

\vspace{-2mm}
\begin{small}
\begin{align}
	\mathbb{E} [\|\mathbf{\tilde{G}}^k_n-\mathbf{G}^k_n \|^2]&\le(1-\theta_n^k)\|\mathbf{G}^k_n \|^2 \nonumber\\
	&\le (1-\theta_n^k)H^2G^2.
\end{align}
\end{small}The compression ratio has two contrasting effects on the training process. 
The larger compression ratio can preserve more information from the original gradients, which reduces the compression error and thus ensures the model accuracy.
However, the communication overhead is still high under these circumstances.
Conversely, the smaller compression ratio will contribute to reducing communication overhead, but it leads to higher compression error and is more likely to deteriorate the model accuracy.
To strike a judicious trade-off between resource overhead and model accuracy, we aim to minimize the compression error under given resource constraints.

\begin{figure*}[t]
	\centering
	\subfigure[LR over MNIST]{
		\centering
		\includegraphics[width=0.23\linewidth]{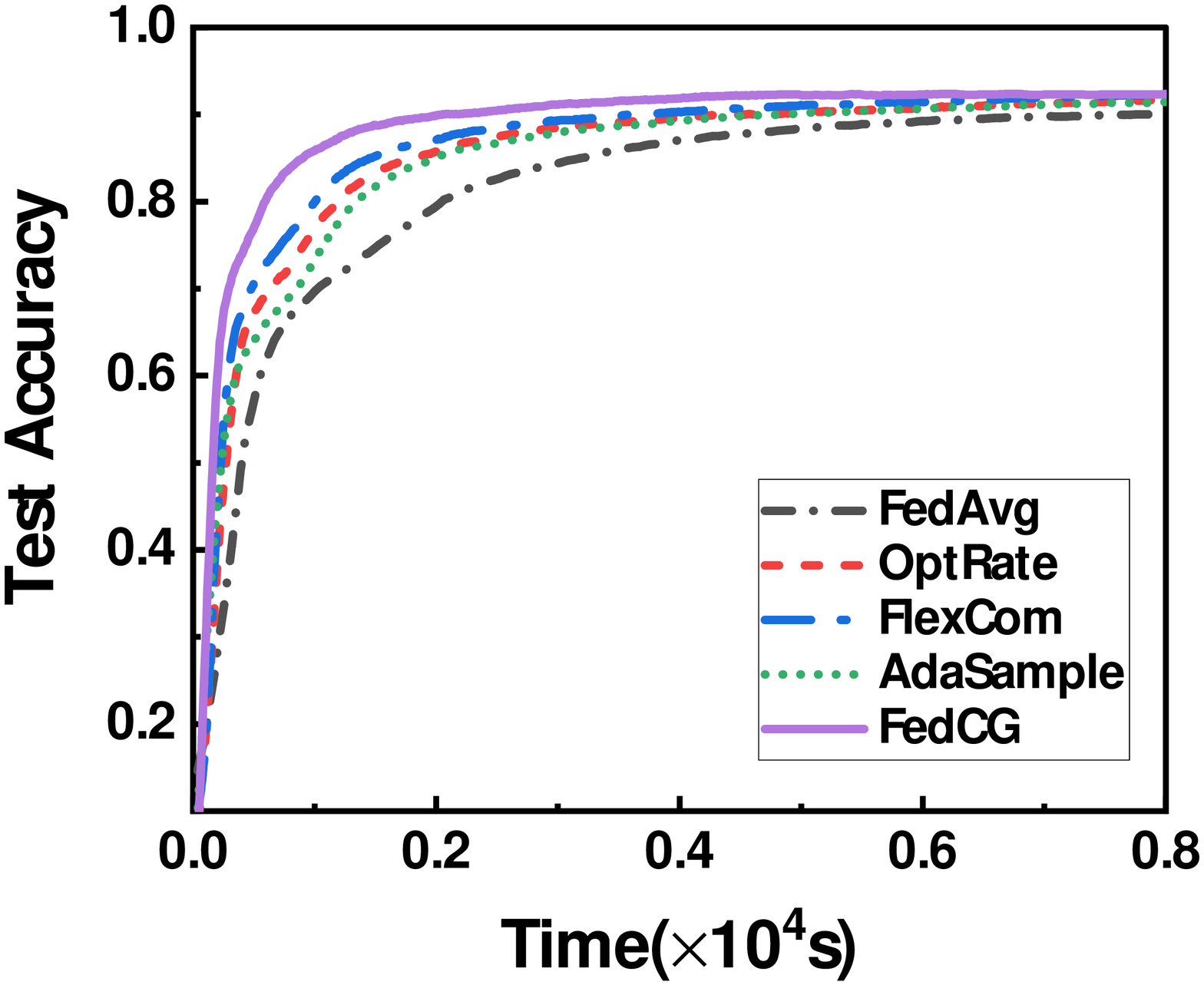}
	}
	\subfigure[AlexNet over CIFAR-10]{
		\centering
		\includegraphics[width=0.23\linewidth]{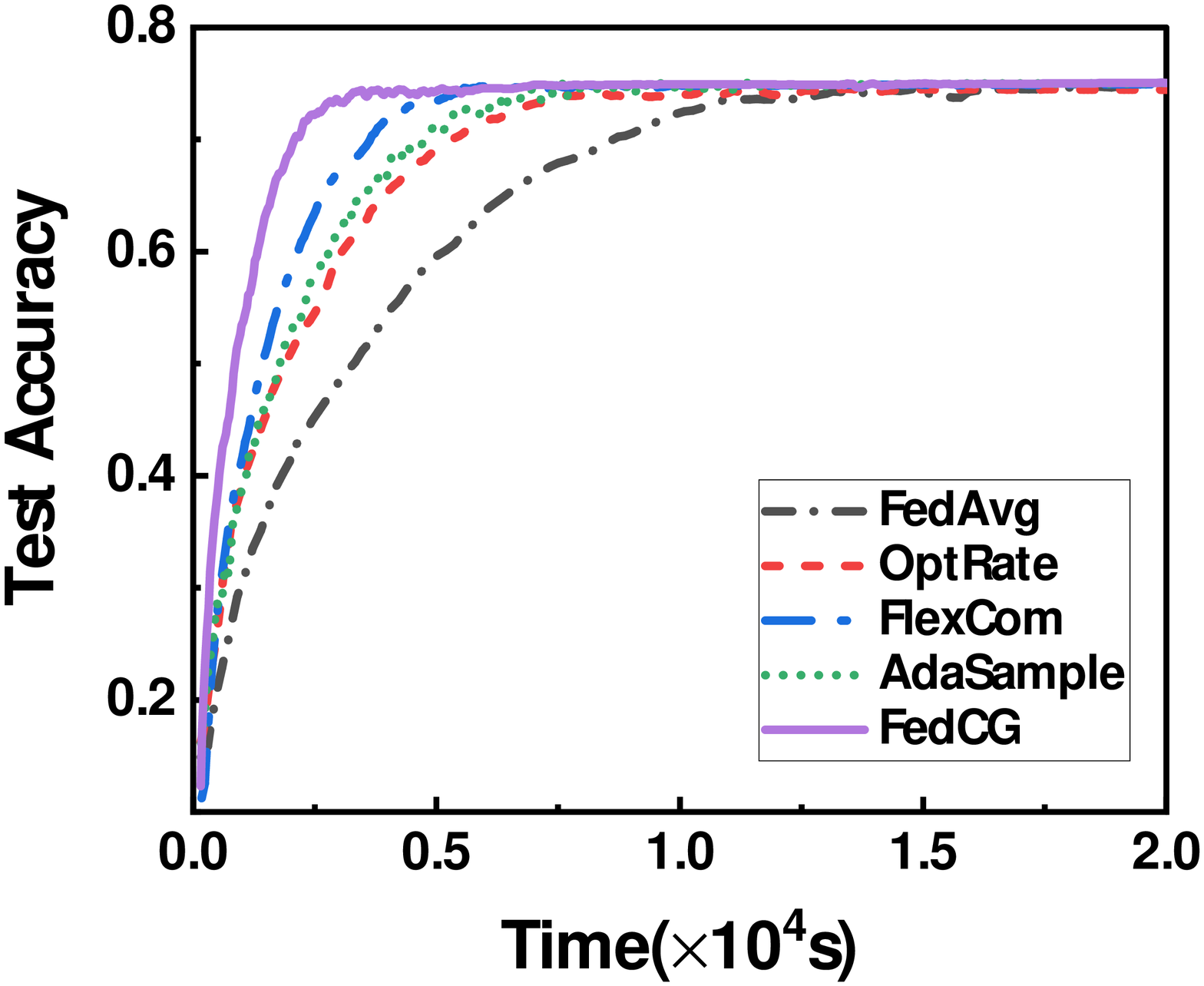}
	}
	\subfigure[ResNet9 over CIFAR-100]{
		\centering
		\includegraphics[width=0.23\linewidth]{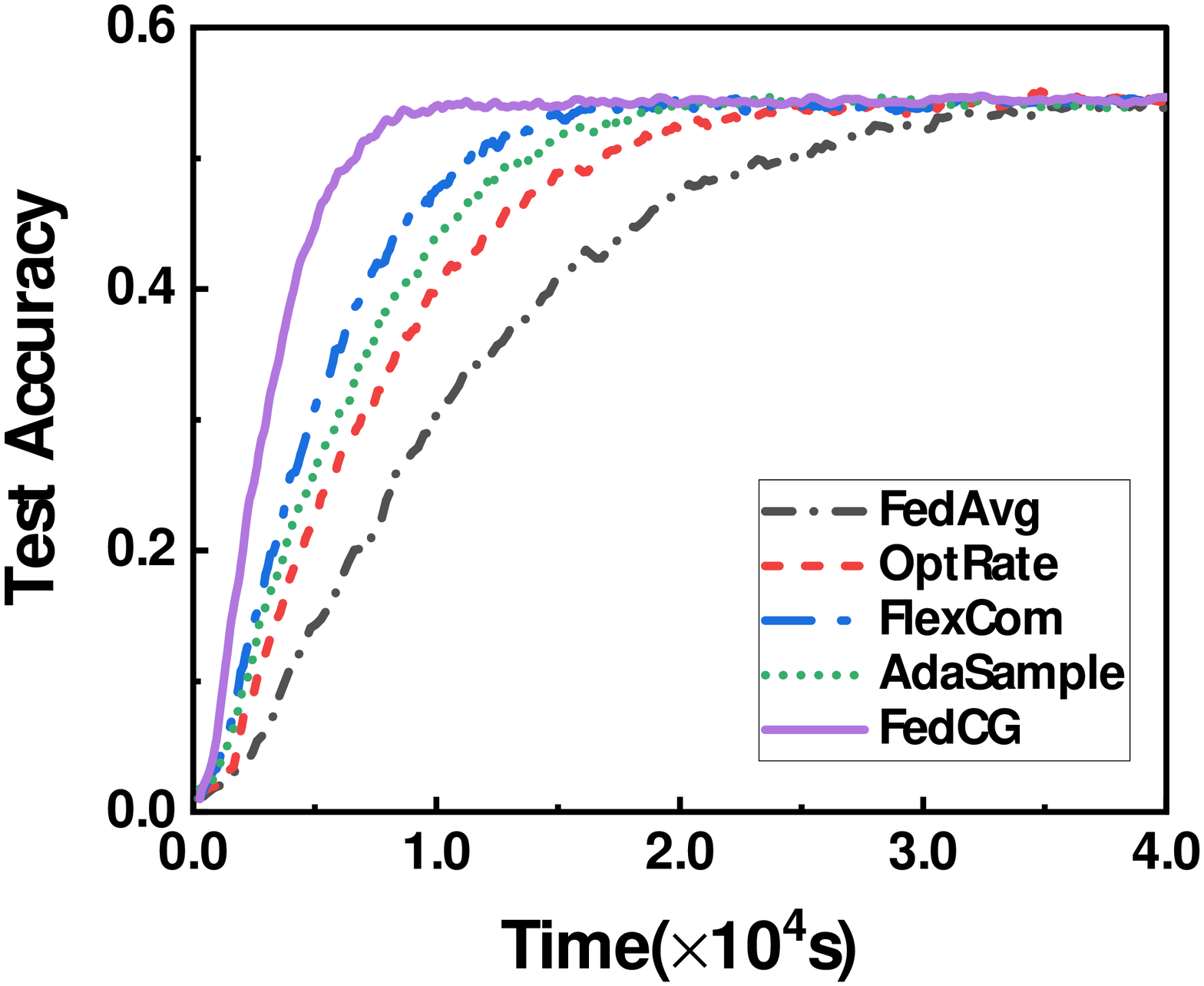}
	}
	\subfigure[ResNet18 over Tiny-ImageNet]{
		\centering
		\includegraphics[width=0.23\linewidth]{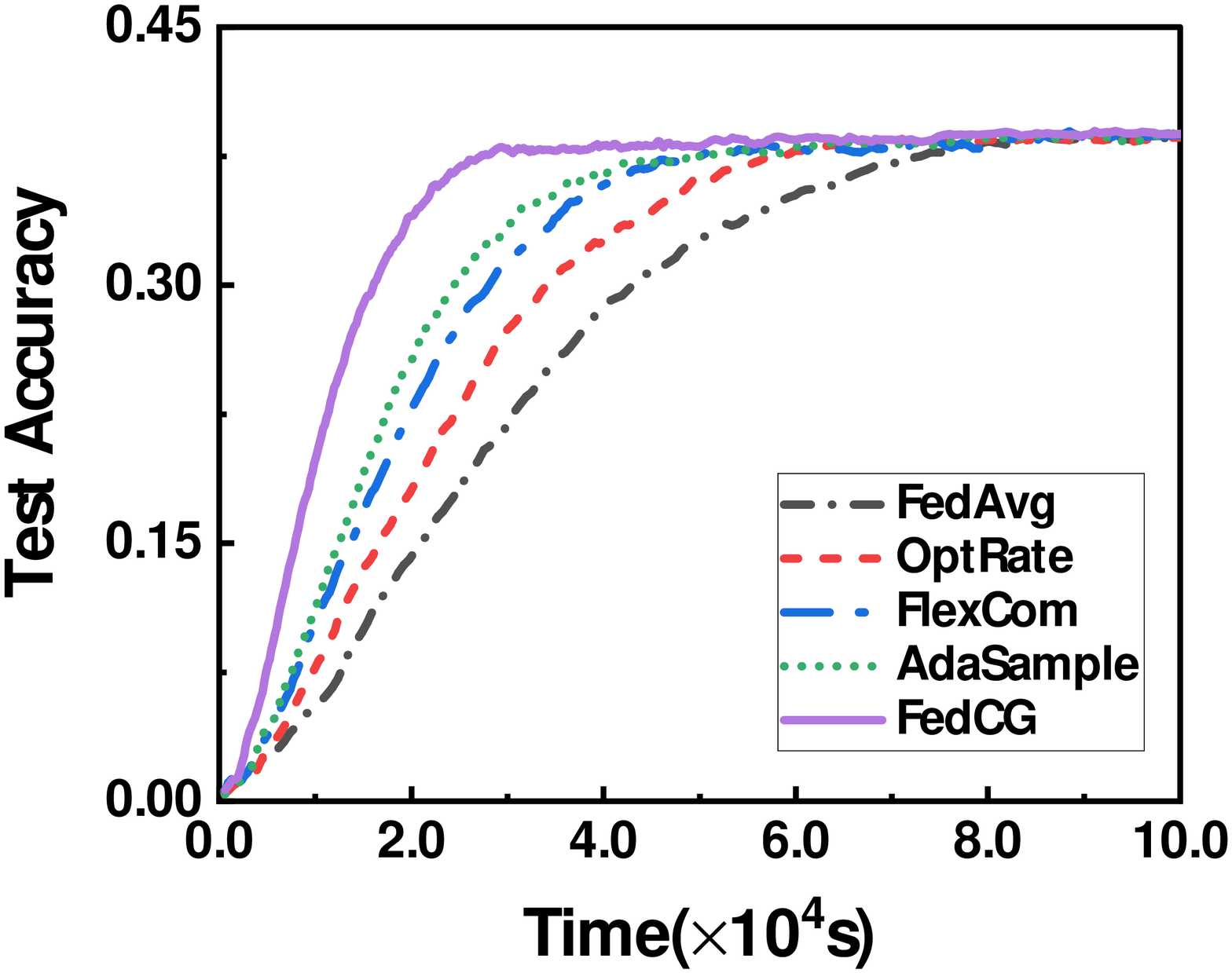}
	}
	\vspace{-1.5mm}
	\caption{Training performance of different methods on the prototype system.}\label{fig:overall}
	\vspace{-2mm}
\end{figure*}

However, the PS requires complete information of the entire training process (\eg, network conditions) to determine optimal compression ratios for selected clients.
Unfortunately, the communication conditions of wireless links are usually time-varying due to network bandwidth reallocation and clients’ mobility, and it is usually impossible to obtain this information in advance.
To overcome the unavailability of future information, we divide the long-term optimization problem into a series of one-shot problems. Given the remaining resources, we online determine the compression ratios for the current round.
Thus, the compression ratios can be continuously adjusted to accommodate system dynamics without requiring future network conditions as prior knowledge. 
Accordingly, compression ratio decision problem at round $k$ is expressed as:
\begin{equation}\small
	\min\ \sum_{n\in\mathcal{M}^k}(1-\theta_n^k)H^2G^2\nonumber
\end{equation}
\begin{equation}\label{equ:ratio}\small
	s.t.\\
	\begin{cases}
		\sum_{i=0}^{k-1}T^i+(K-k)T^k<T  \\
		0<\theta_n^k\leq1, &\forall n, \forall k \\
	\end{cases}
\vspace{-1mm}
\end{equation}
Since the above optimization problem is a linear programming (LP) problem, it can be optimally solved using LP solver (\eg, PuLP \cite{mitchell2011pulp}).
By solving the problem in Eq. (\ref{equ:ratio}), FedCG assigns different compression ratios to the selected clients according to their heterogeneous and time-varying capabilities.
Consequently, these clients adaptively compress and upload the gradients under resource constraints, preventing the ones with poor capabilities from becoming the bottleneck of FL.

%% file: content/evaluation.tex
\subsection{Experimental Setup}
\textbf{Experimental Platforms.} We evaluate the performance of FedCG in both a physical testbed and a simulated environment.
The prototype system helps us capture real-world resource overhead (\eg, time and traffic consumption) and the simulation system is used to evaluate larger-scale FL scenarios with manipulative parameters.
\textit{(1) Testbed settings:} The hardware prototype system consists of an AMAX deep learning workstation as the PS and 30 commercial embedded devices as the clients. Specifically, the workstation is carrying one 8-core Intel Xeon
CPU and 4 NVIDIA GeForce RTX 2080Ti GPUs. The clients include 10 NVIDIA Jetson AGX devices, 10 NVIDIA Jetson NX devices, and 10 NVIDIA Jetson TX2 devices\footnote{https://developer.nvidia.com/embedded/community/support-resources}, which reflect capability heterogeneity of clients. All devices are interconnected via a commercial WiFi router and we develop a TCP-based socket interface for communication between the PS and clients.
\textit{(2) Simulation settings:} We simulate an FL system with 100 virtual clients (each is implemented as a process in the system). To reflect heterogeneous and dynamic network conditions, we fluctuate each client's inbound bandwidth between 10Mb/s and 20Mb/s. Considering that the outbound bandwidth is typically smaller than the inbound bandwidth in a typical WAN, we set it to fluctuate between 1Mb/s and 5Mb/s \cite{wang2022accelerating} and randomly change it every round. For computation heterogeneity, the time of one local iteration follows a Gaussian distribution whose mean and variance are from the measurements of the prototype.

\textbf{Datasets and Models.} We conduct the experiments over four datasets (\ie, MNIST, CIFAR-10, CIFAR-100, and Tiny-ImageNet), which represent a large variety of the small, middle and large training tasks in practical FL scenarios.
We adopt the \textit{convex} logistic regression (LR) model for MNIST and \textit{non-convex} deep neural networks (\eg, AlexNet, ResNet9, and ResNet18) for the other three datasets.

\textbf{Data Distribution.} To simulate various degrees of statistical data heterogeneity, we adopt two different non-IID partition schemes, \ie, latent Dirichlet allocation (LDA) and skewed label, which are widely used in previous works \cite{wang2020optimizing,wang2022accelerating,luo2021tackling}.
\textit{(1) LDA for MNIST and CIFAR-10:} $\psi$ ($\psi=$ 0.2, 0.4, 0.6, and 0.8) of the data on each client belong to one class and the remaining $1-\psi$ samples belong to other classes.
\textit{(2) Skewed label for CIFAR-100 and Tiny-ImageNet:} Each client lacks $\psi$ classes of data samples, where $\psi=$ 20, 40, and 60 for CIFAR-100, and $\psi=$ 40, 80, and 120 for Tiny-ImageNet.
In particular, we use $\psi=0$ to denote IID data. Except for the experiments on non-IID data, we shuffle the data and uniformly divide them among all clients.

\textbf{Benchmarks.} We compare the proposed framework with four benchmarks.
\textit{(1) FedAvg}\cite{mcmahan2017communication} selects clients uniformly at random and exchanges the entire models between the PS and selected clients.
\textit{(2) OptRate}\cite{cui2021optimal} adopts compression to reduce communication overhead and determines identical compression rates for clients at each round to seek the trade-off between overhead and accuracy.
\textit{(3) FlexCom}\cite{li2021talk} enables flexible compression control and allows clients to compress the gradients to different levels considering the heterogeneity in communication capabilities.
\textit{(4) AdaSample}\cite{luo2021tackling} optimizes client sampling probabilities to tackle both system and statistical heterogeneity so as to minimize FL completion time.
Since we concentrate on improving the training efficiency of FL regarding resource constraints, training models to achieve state-of-the-art accuracy is beyond the scope of this work.
Unless otherwise specified, we select $M=10$ clients to participate in training and the clients perform $H = 50$ local iterations at each round.

\subsection{Testbed Results}
\textbf{Training Performance.}
We first compare the training performance of FedCG and other methods on the prototype system.
The accuracy results with respect to training time are presented in Fig. \ref{fig:overall}.
We observe that FedCG achieves a comparable accuracy and converges much faster than the other methods for all four datasets.
Compared to the benchmarks, FedCG can provide up to 3.8$\times$ speedup for LR over MNIST, 4.9$\times$ speedup for AlexNet over CIFAR-10, 3.5$\times$ speedup for ResNet9 over CIFAR-100, and 2.7$\times$ speedup for ResNet18 over Tiny-ImageNet.
Furthermore, the accuracy of FedCG always surpasses the other benchmarks after a given time.
In particular, our framework achieves 74.94\% accuracy after training AlexNet over CIFAR-10 for 10,000s while that of FedAvg, OptRate, FlexCom and AdaSample is 72.49\%, 74.01\%, 74.83\% and 74.81\%, respectively.
These results demonstrate the advantages of FedCG for both convex and non-convex learning tasks.

\begin{table}[]
	\centering
	\renewcommand\arraystretch{1.5}
	\caption{Resource overhead of different methods to achieve the target accuracy.}\label{tb:overhead}
	\resizebox{0.5\textwidth}{!}{
		\begin{tabular}{|c|c|c|c|c|c|c|}
			\hline
			Datasets                                                                                 & Metrics     & FedAvg  & OptRate & FlexCom & AdaSample & \textbf{FedCG}   \\ \hline
			\multirow{2}{*}{\begin{tabular}[c]{@{}c@{}}MNIST\\[-1ex] (Acc=90\%)\end{tabular}}              & Time(s)     & 7754.2  & 4545.1  & 3647.2  & 4818.1    & \textbf{2057.5}  \\ \cline{2-7} 
			& Traffic(MB) & 4368.1  & 1556.4  & 898.9   & 4072.5    & \textbf{874.7}  \\ \hline
			\multirow{2}{*}{\begin{tabular}[c]{@{}c@{}}CIFAR-10\\[-1ex] (Acc=74\%)\end{tabular}}      & Time(s)     & 15932.1 & 9696.9  & 5334.3  & 6967.6    & \textbf{3261.8}  \\ \cline{2-7} 
			& Traffic(MB) & 15198.6 & 6192.8  & 2674.2  & 11983.6   & \textbf{2480.3}  \\ \hline
			\multirow{2}{*}{\begin{tabular}[c]{@{}c@{}}CIFAR-100\\[-1ex] (Acc=54\%)\end{tabular}}     & Time(s)     & 35047.9 & 24520.5 & 17726.2 & 19722.6   & \textbf{10068.7} \\ \cline{2-7} 
			& Traffic(MB) & 32550.9 & 15581.5 & 8583.2  & 34569.6   & \textbf{8401.7}  \\ \hline
			\multirow{2}{*}{\begin{tabular}[c]{@{}c@{}}Tiny-ImageNet\\[-1ex] (Acc=37\%)\end{tabular}} & Time(s)     & 68931.7 & 54918.1 & 45543.9 & 41717.9   & \textbf{25614.4} \\ \cline{2-7} 
			& Traffic(MB) & 53271.6 & 33242.0 & 21107.4 & 55849.3   & \textbf{19958.7} \\ \hline
	\end{tabular}}
	\vspace{-2mm}
\end{table}

\textbf{Resource Overhead.}
To validate the efficiency of FedCG, we record the resource overhead of different methods when they attain the target accuracy in Table \ref{tb:overhead}, including completion time and traffic consumption. Note that the target accuracy is set as the accuracy that all methods can achieve.
As summarized in Table \ref{tb:overhead}, compared with the benchmarks, FedCG reduces the training time by 55.4\% and network traffic by 50.3\% on average for reaching the same accuracy.
The reasons for such superior performance are as follows.
Compared to FedAvg and AdaSample, the solutions with model compression (\ie, OptRate, FlexCom and FedCG) can save much more network traffic. However, OptRate assigns identical compression ratios to heterogeneous clients, which exacerbates the straggler effect.
Although FlexCom and FedCG achieve similar traffic consumption by assigning different compression ratios to clients, FlexCom still produces a long training completion time without considering the computation heterogeneity.
In addition, we note that optimizing sampling probabilities like AdaSample can improve training efficiency to a certain extent, but cannot essentially address heterogeneity challenges since the clients with large capability gaps may participate in FL at the same round.
By contrast, with the assistance of adaptive client selection and gradient compression, FedCG brings significant savings in both time and traffic consumption, thereby effectively accelerating the training process of FL.

\begin{figure*}[t]
	\centering
	\subfigure[LR over MNIST]{
		\centering
		\includegraphics[width=0.23\linewidth]{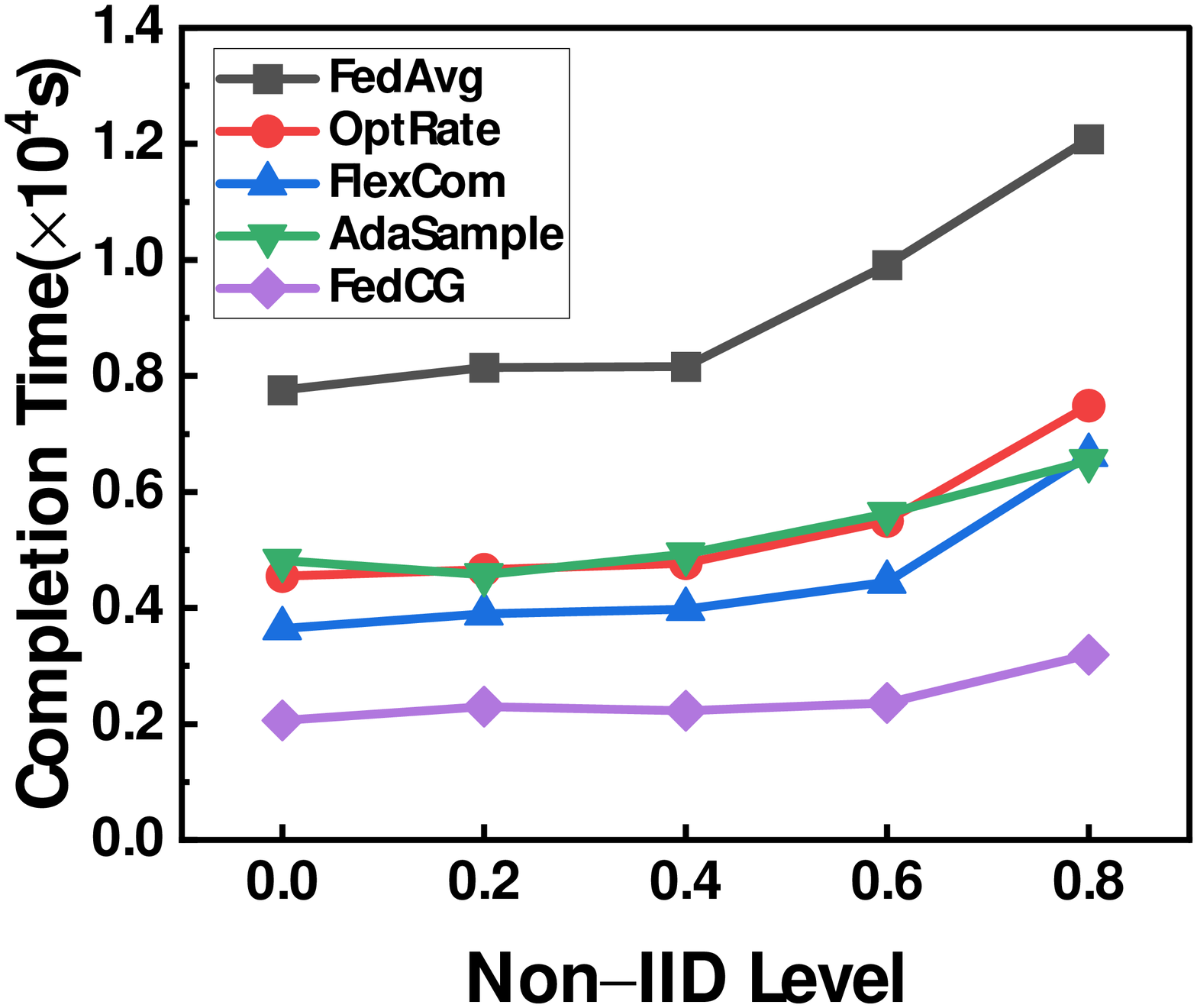}
	}
	\subfigure[AlexNet over CIFAR-10]{
		\centering
		\includegraphics[width=0.23\linewidth]{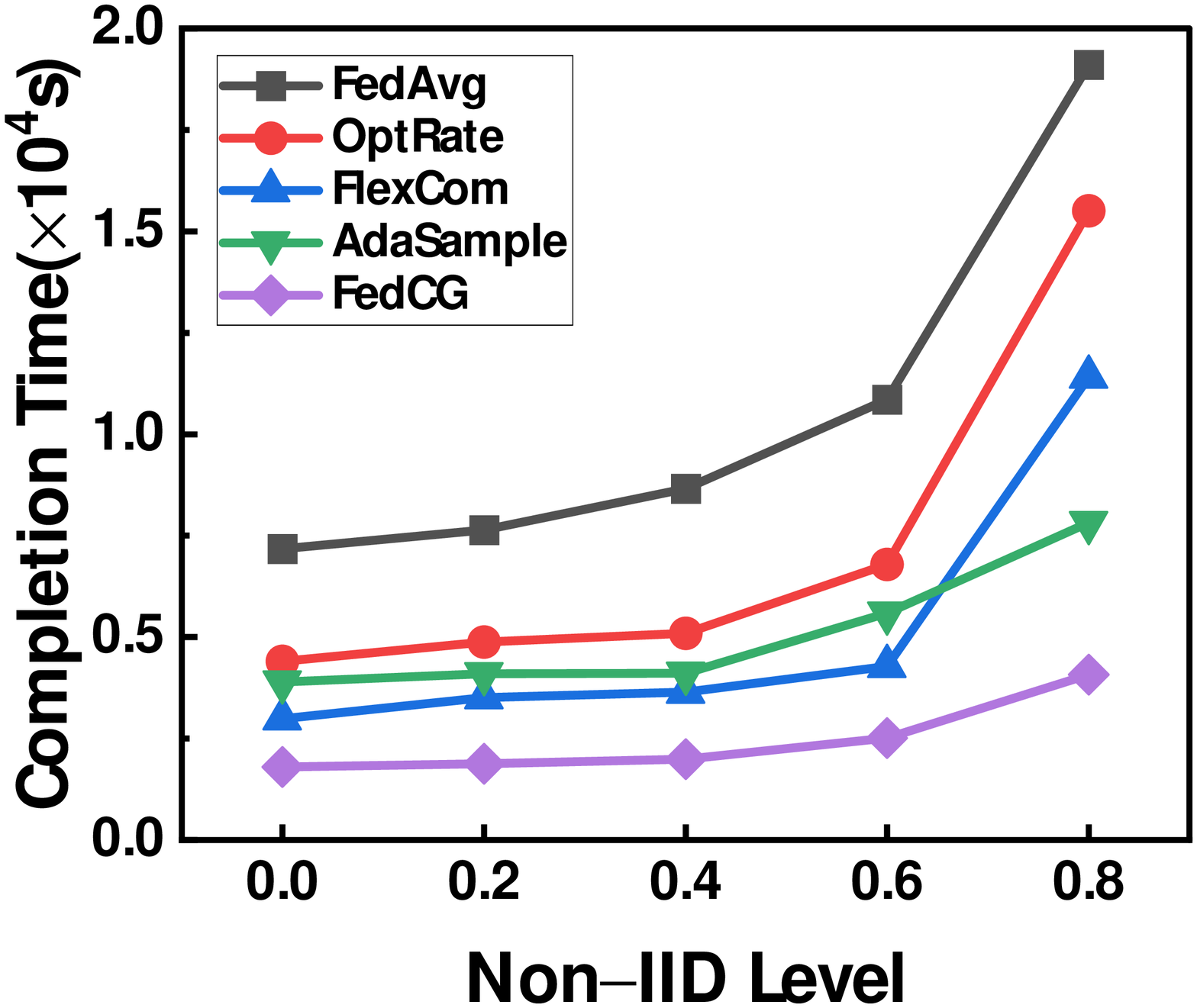}
	}
	\subfigure[ResNet9 over CIFAR-100]{
		\centering
		\includegraphics[width=0.23\linewidth]{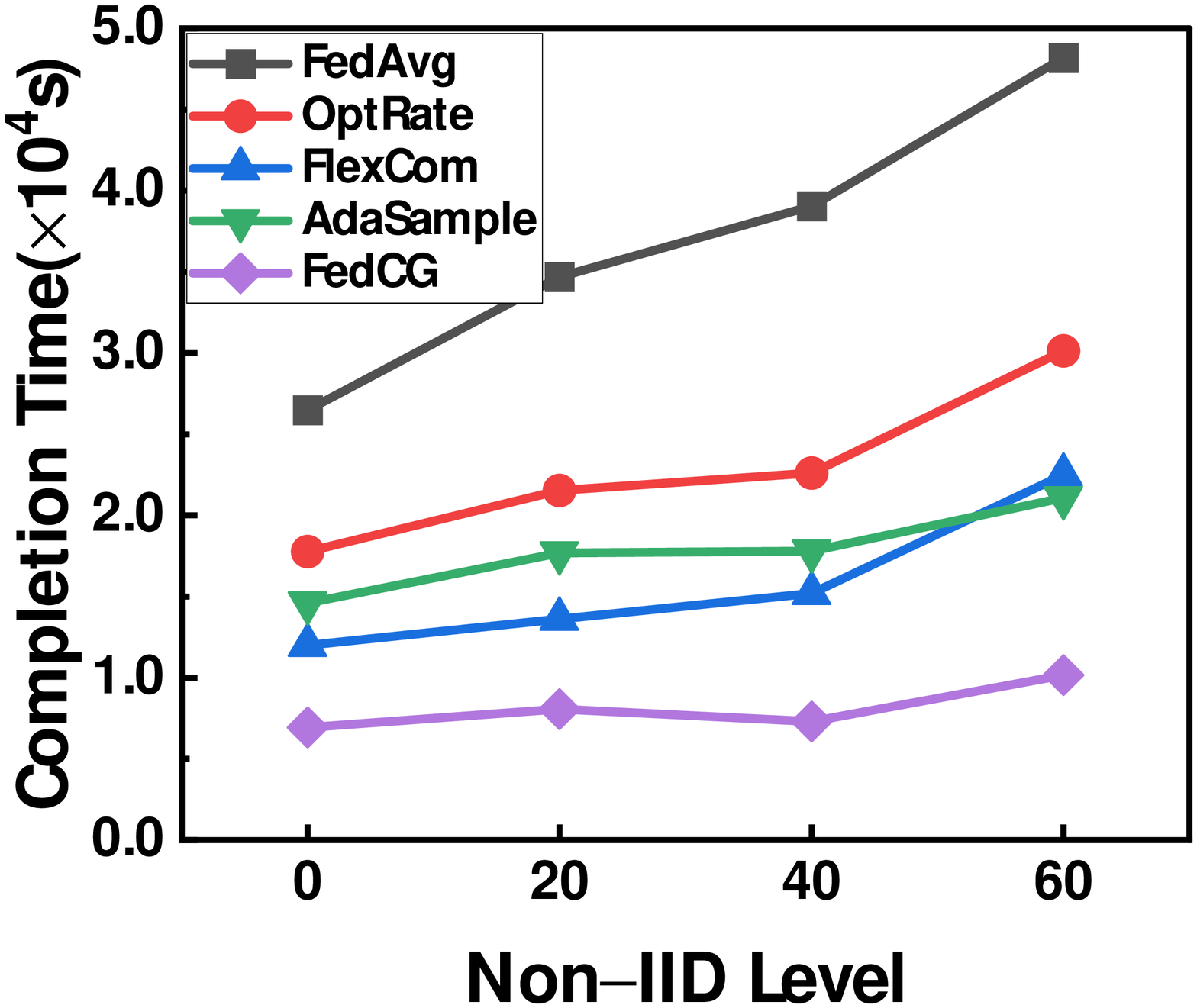}
	}
	\subfigure[ResNet18 over Tiny-ImageNet]{
		\centering
		\includegraphics[width=0.23\linewidth]{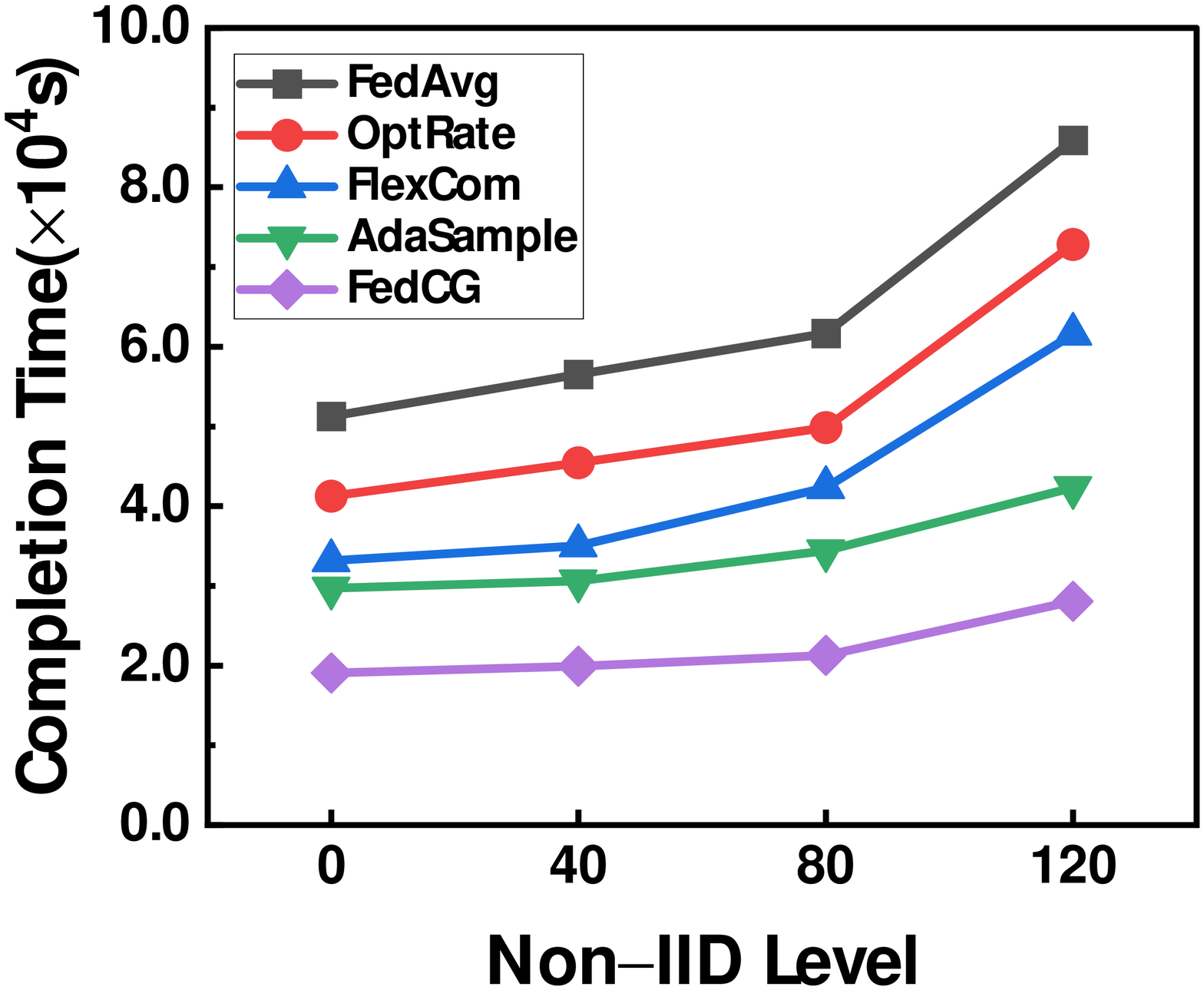}
	}
	\vspace{-2mm}
	\caption{Completion time under different levels of non-IID data.}\label{fig:noniid}
	\vspace{-4mm}
\end{figure*}

\begin{figure}[t]
	\centering
	\subfigure[LR over MNIST]{
		\centering
		\includegraphics[width=0.465\linewidth]{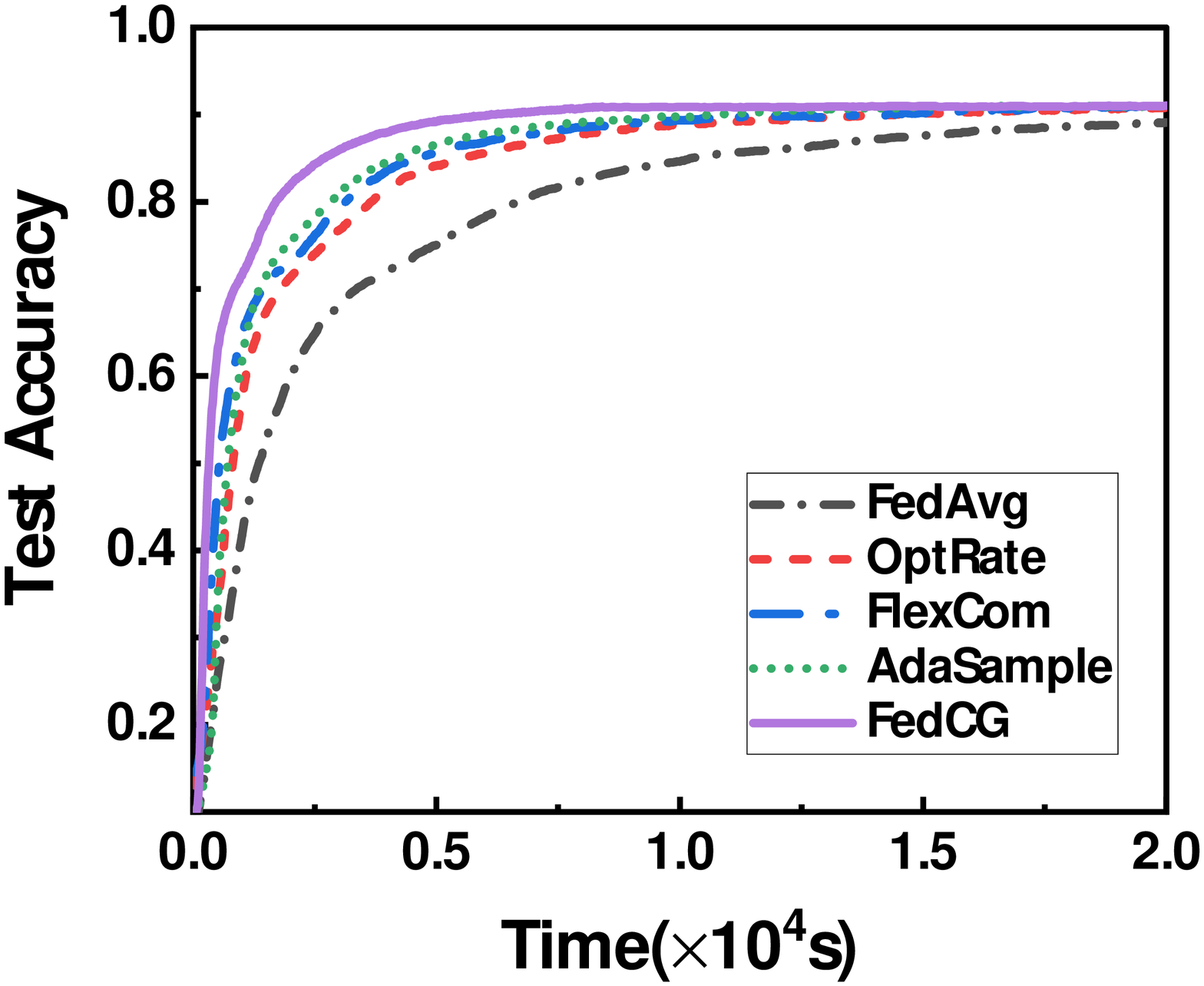}
	}
	\subfigure[AlexNet over CIFAR-10]{
		\centering
		\includegraphics[width=0.465\linewidth]{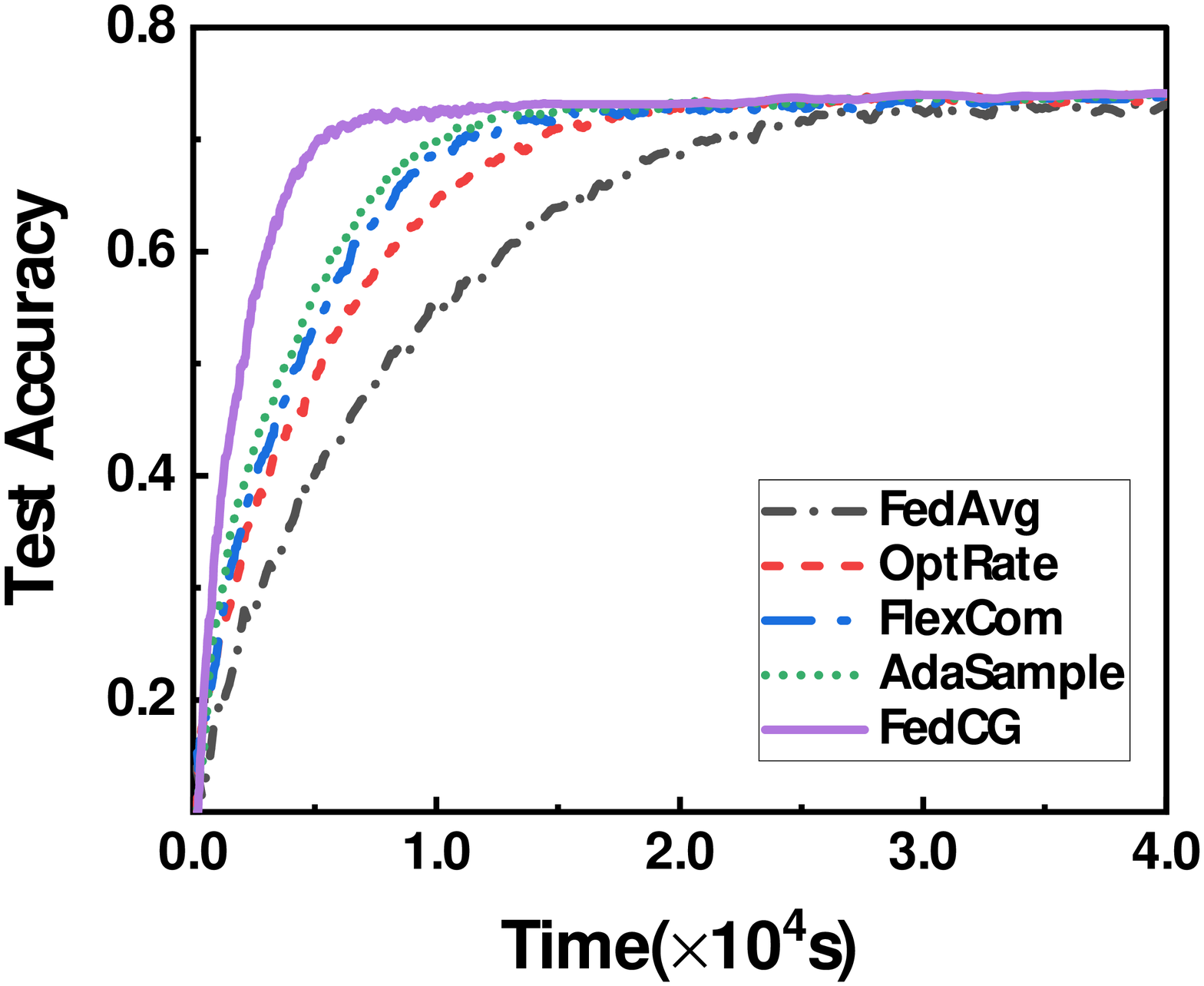}
	}
	\vspace{-1.5mm}
	\caption{Training performance in dynamic and heterogeneous simulation environments.}\label{fig:simulation}
	\vspace{-3mm}
\end{figure}

\begin{figure}[t]
	\vspace{-1mm}
	\centering
	\subfigure[Network Traffic]{
		\centering
		\includegraphics[width=0.465\linewidth]{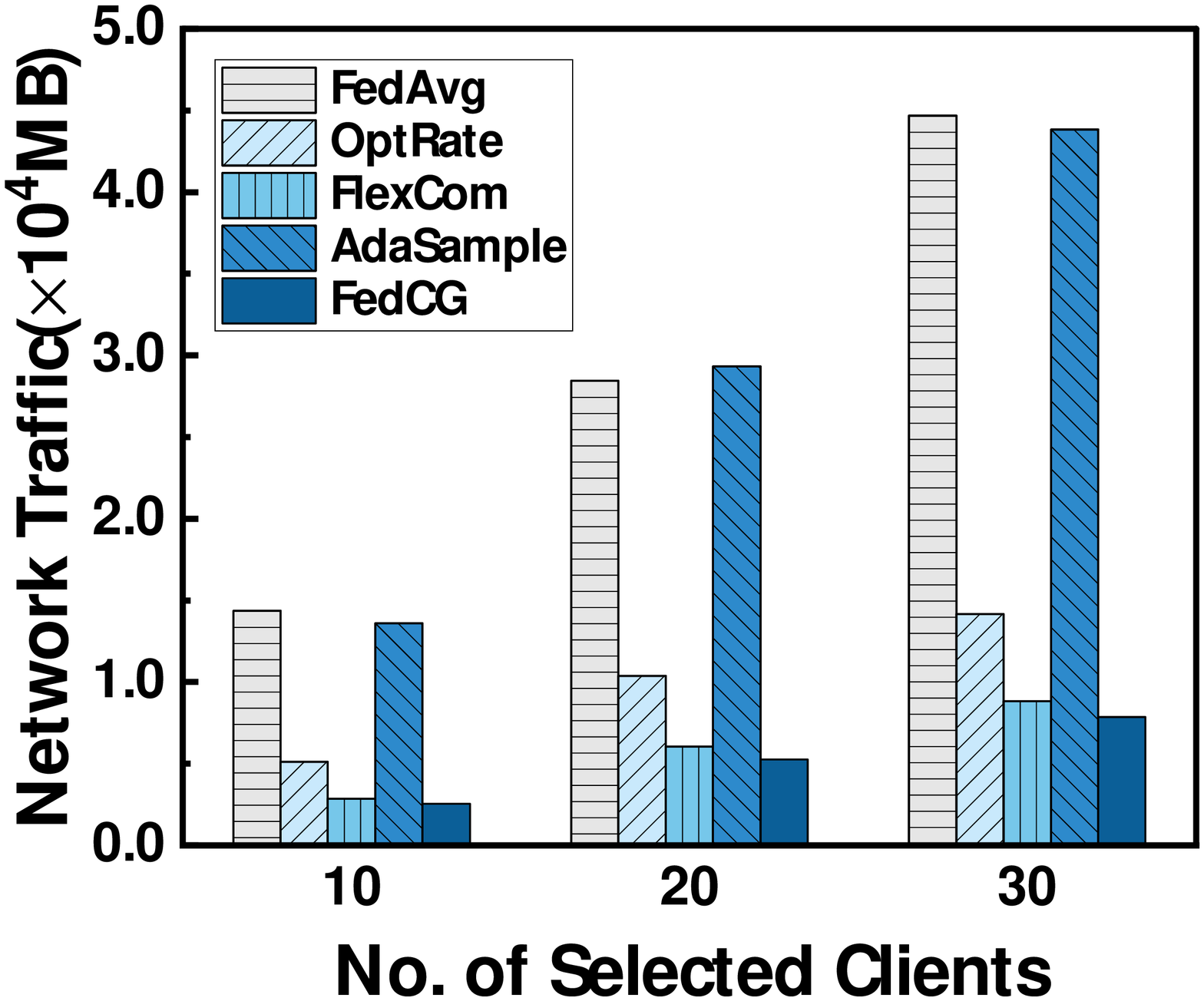}\label{fig:network_traffic}
	}
	\subfigure[Algorithm Overhead]{
		\centering
		\includegraphics[width=0.465\linewidth]{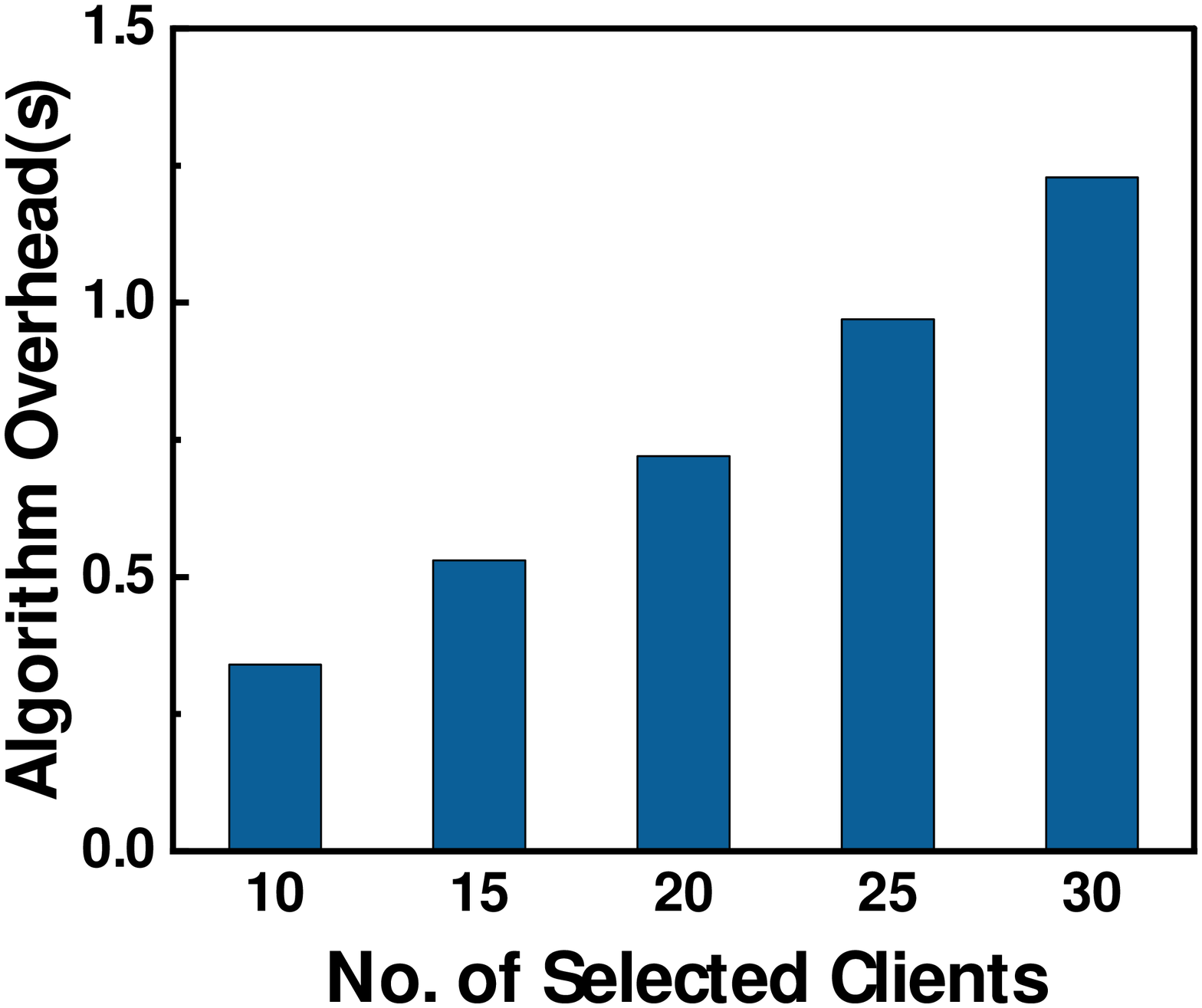}\label{fig:algorithm_overhead}
	}
	\vspace{-1.5mm}
	\caption{Effect of the number of selected clients on network traffic and algorithm overhead.}
	\vspace{-3mm}
\end{figure}

\textbf{Effect of Non-IID Data.}
We proceed to investigate how our proposed framework performs under statistical heterogeneity.
Fig. \ref{fig:noniid} depicts the required time for FedCG and benchmarks to reach the target accuracy under different levels of non-IID data.
We set the target accuracy of LR, AlexNet, ResNet9 and ResNet18 as 90\%, 67\%, 51\% and 33\%, respectively.
As shown in Fig. \ref{fig:noniid}, all methods suffer from performance degradation with the increasing skewness of data distribution.
Nevertheless, FedCG only has the slightest increase in completion time compared to the other benchmarks and exhibits robustness to non-IID data.
The advantage of FedCG is attributed to diverse client selection which increases the impact of under-represented clients, thereby promoting fairness and reducing the bias introduced by non-IID data.
In addition, the savings of resource overhead in FedCG further enlarge the performance gap as the non-IID level increases.
The speedup provided by FedCG is 4.0$\times$, 4.1$\times$, 4.3$\times$, 4.3$\times$, and 4.7$\times$ with non-IID level varying from 0 to 0.8 for AlexNet over CIFAR-10, indicating the effectiveness of FedCG for data heterogeneity.

\subsection{Simulation Results}
\textbf{Dynamic and Heterogeneous Environments.}
To evaluate the performance of FedCG in large-scale FL scenarios, we conduct our experiments by simulations with 100 clients.
Fig. \ref{fig:simulation} plots the accuracy curve of different methods in dynamic and heterogeneous environments.
We find that our proposed framework still substantially outperforms the other benchmarks in large-scale FL scenarios and exhibits faster convergence without sacrificing accuracy.
For instance, FedCG takes 5,170s to achieve 70\% accuracy for AlexNet over CIFAR-10, while the completion time of FedAvg, OptRate, FlexCom and AdaSample are 22,009s, 14,225s, 11,129s and 10,392s, respectively.
The corresponding speedups are 4.3$\times$, 2.8$\times$, 2.2$\times$ and 2.0$\times$.
Such performance gain of FedCG is rooted in appropriate client selection strategy and different compression ratios. This not only excludes the clients with poor capabilities from training but also allows each selected client to transmit compressed gradients fitting its capabilities.
Moreover, the decisions including client subset and compression ratio are continuously adjusted during training to adapt to the time-varying capabilities of clients.
The above simulation results strongly verify the usability of our design in highly dynamic and heterogeneous environments.

\begin{figure}[t]
	\centering
	\subfigure[AlexNet over CIFAR-10]{
		\centering
		\includegraphics[width=0.465\linewidth]{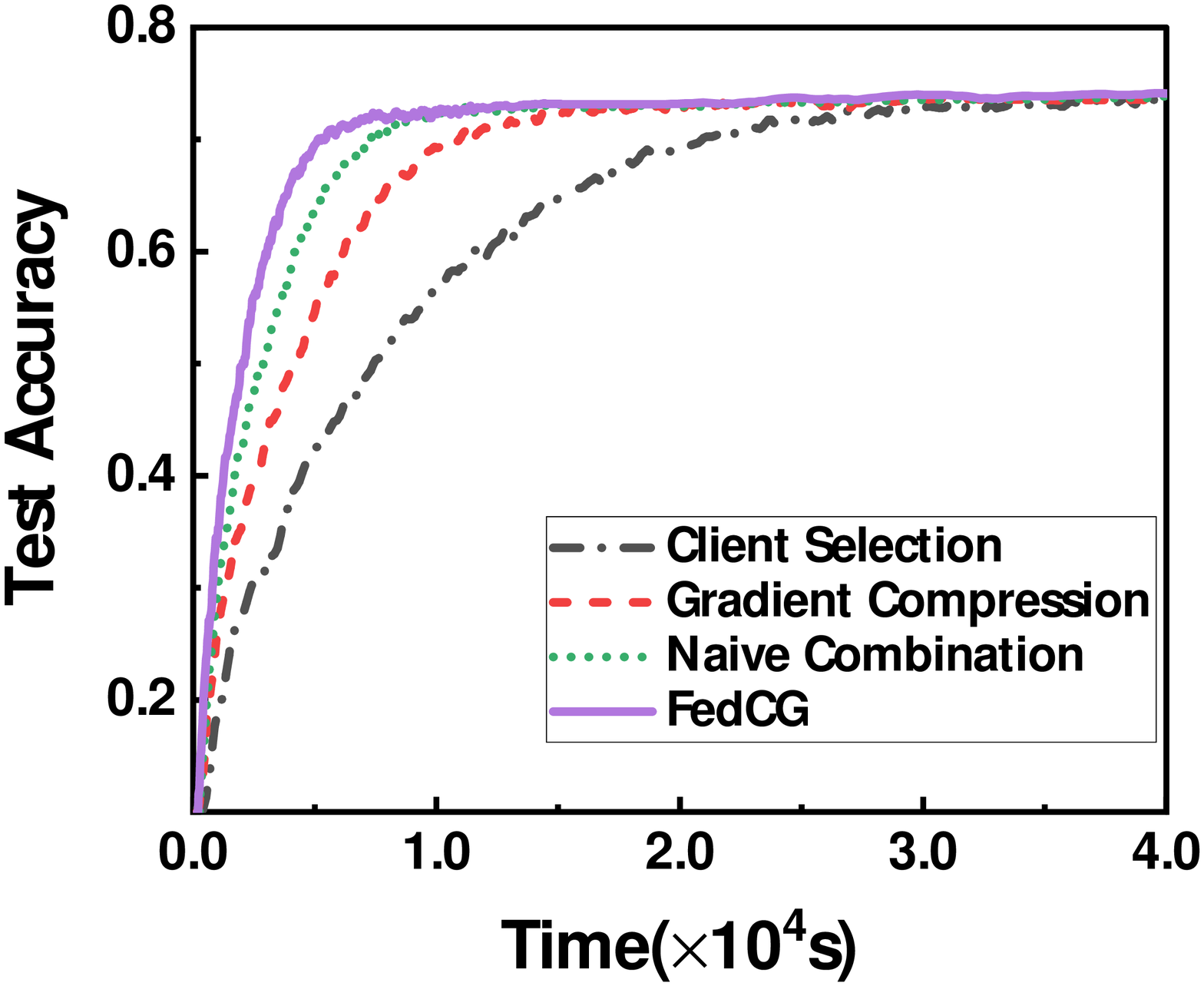}
	}
	\subfigure[ResNet9 over CIFAR-100]{
		\centering
		\includegraphics[width=0.465\linewidth]{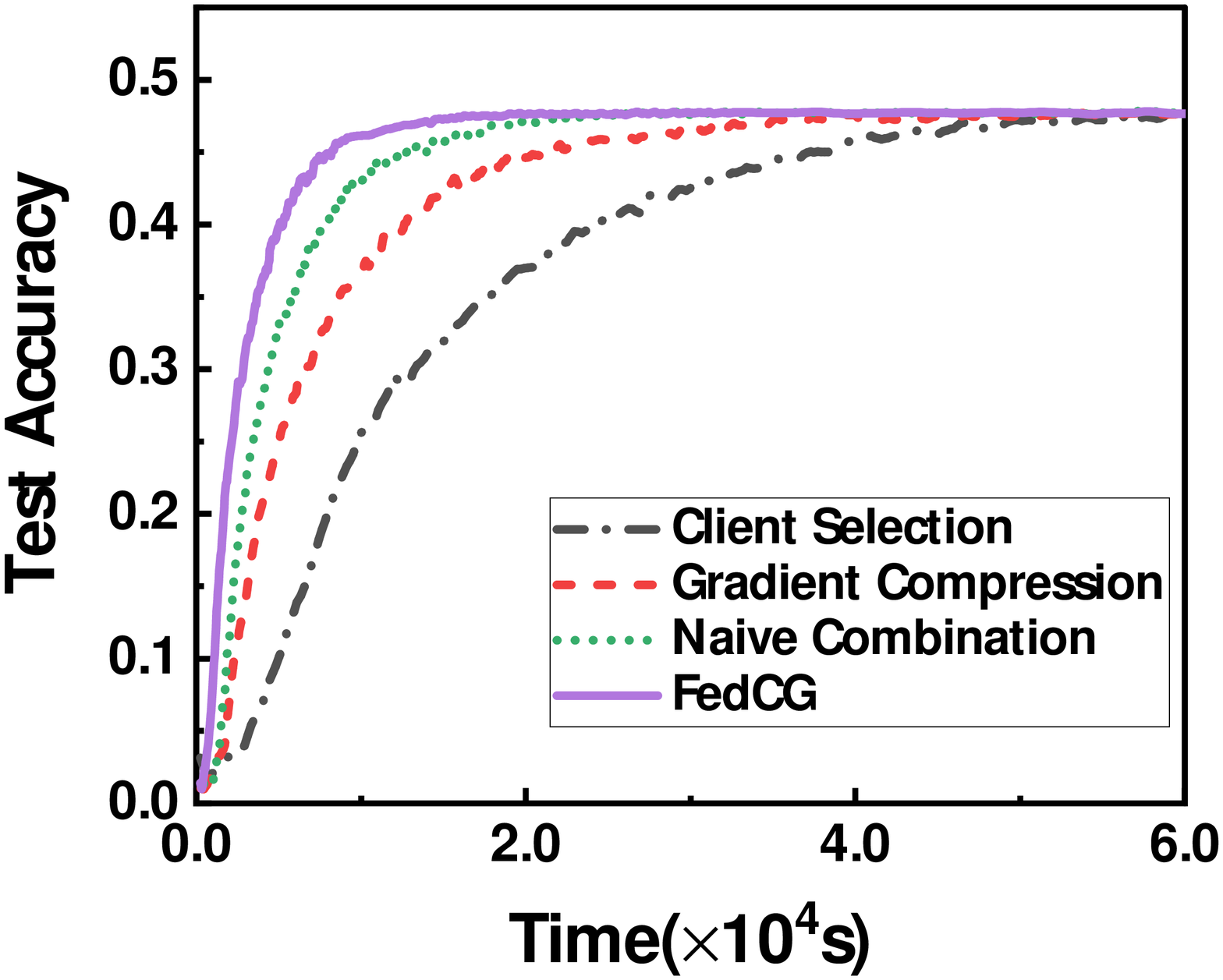}
	}
	\vspace{-1.5mm}
	\caption{Training performance of independent decision, naive combination and joint optimization.}\label{fig:joint_optimization}
	\vspace{-4mm}
\end{figure}

\textbf{Varying the Number of Selected Clients.}
We further conduct the simulation experiments to analyze the influence of the number of selected clients (\ie, $M$) on the training efficiency.
Firstly, we compare the traffic consumption of different methods for achieving the target accuracy (\eg, 90\%) when the number of selected clients increases from 10 to 30.
The results for LR over MNIST are shown in Fig. \ref{fig:network_traffic}.
Apparently, network traffic of all methods increases gradually with $M$ ranging from 10 to 30.
This is expected because more clients participate in training at each round, which will consume more communication resources to transmit model updates.
Nevertheless, FedCG outperforms other methods under different numbers of selected clients and reduces network traffic consumption by about 11.2-82.4\% compared with the four benchmarks.
Secondly, we measure the decision overhead of the proposed joint optimization algorithm with various numbers of selected clients, which is illustrated in Fig. \ref{fig:algorithm_overhead}.
Although the algorithm overhead becomes larger as $M$ increases, the maximum overhead is only 1.2s, which is much smaller than the FL training and transmission time (\eg, hundreds of seconds) and thus can be ignored.
These results suggest that the iterative optimization process of the proposed algorithm incurs a small decision overhead and will not hinder the practical deployment of FedCG in FL.

\textbf{Necessity of Joint Optimization.}
Instead of simple combination, FedCG aims to achieve efficient FL by joint optimization of client selection and gradient compression.
To indicate the importance of the proposed joint optimization algorithm, we compare the training performance of independent decision, naive combination and FedCG.
As shown in Fig. \ref{fig:joint_optimization}, it is clear that FedCG consistently converges faster than the other three methods without loss of accuracy.
Our framework provides 1.4-4.1$\times$ speedup to reach target accuracy (\eg, 70\%) for AlexNet over CIFAR-10 and 1.5-4.5$\times$ speedup to reach target accuracy (\eg, 45\%) for ResNet9 over CIFAR-100.
The explanation for this phenomenon is that client selection and compression ratio decision are tightly coupled. Consequently, neither independent decision nor naive combination can handle network dynamics and client heterogeneity well, thus negatively affecting the convergence performance.
FedCG overcomes this issue by proposing an iteration-based algorithm and demonstrates impressive performance improvement, which emphasizes the necessity of joint optimization process.

%% file: content/conclusion.tex
In this paper, we propose a novel framework, called FedCG, to achieve efficient FL with adaptive client selection and gradient compression.
Specifically, FedCG selects a diverse set of clients and assigns different compression ratios to selected clients considering their heterogeneous and time-varying capabilities.
We jointly optimize client selection and compression ratio decision, which address the challenges of FL on communication efficiency, network dynamics and client heterogeneity.
Experimental results demonstrate the advantages and effectiveness of the proposed framework.

\section*{Acknowledgment}
This article is supported in part by the National Key Research and Development Program of China (Grant No. 2021YFB3301501); in part by the National Science Foundation of China (NSFC) under Grants 62132019, 62102391 and 61936015; in part by the Jiangsu Province Science Foundation for Youths (Grant No. BK20210122).